\documentclass[a4paper,11pt]{article}
\pdfoutput=1 

\usepackage{jheppub} 

\usepackage[T1]{fontenc} 

\usepackage{amsmath,amssymb,amsthm,graphicx}
\usepackage[dvipsnames]{xcolor}
\usepackage{multirow}
\usepackage{adjustbox}
\usepackage{comment}
\usepackage{enumerate}
\usepackage{hyperref}
\usepackage{cleveref}
\usepackage{bm}
\usepackage{subcaption}
\usepackage{lscape}
\usepackage{multirow}

\newtheorem{theorem}{Theorem}
\newtheorem{lemma}[theorem]{Lemma}

\begin{document}

\hfill SLAC-PUB-17780

\title{
The Phase Space Distance Between Collider Events
}

\author[a]{Tianji Cai,}
\author[b]{Junyi Cheng,}
\author[c,d]{Nathaniel Craig,}
\author[c]{Giacomo Koszegi,}
\author[e]{and Andrew J.~Larkoski}

\affiliation[a]{SLAC National Accelerator Laboratory, Stanford University, Stanford, CA 94309, USA}
\affiliation[b]{Department of Physics, Harvard University, Cambridge, MA 02138, USA}
\affiliation[c]{Department of Physics, University of California, Santa Barbara, CA 93106, USA}
\affiliation[d]{Kavli Institute for Theoretical Physics, Santa Barbara, CA 93106, USA}
\affiliation[e]{Department of Physics and Astronomy, University of California, Los Angeles, CA 90095, USA\\
Mani L. Bhaumik Institute for Theoretical Physics, University of California, Los Angeles, CA 90095, USA}

\emailAdd{tianji@slac.stanford.edu}
\emailAdd{junyi\_cheng@g.harvard.edu}
\emailAdd{ncraig@ucsb.edu}
\emailAdd{koszegi@ucsb.edu}
\emailAdd{larkoski@physics.ucla.edu}

\abstract{
How can one fully harness the power of physics encoded in relativistic $N$-body phase space? Topologically, phase space is isomorphic to the product space of a simplex and a hypersphere and can be equipped with explicit coordinates and a Riemannian metric.  This natural structure that scaffolds the space on which all collider physics events live opens up new directions for machine learning applications and implementation.  Here we present a detailed construction of the phase space manifold and its differential line element, identifying particle ordering prescriptions that ensure that the metric satisfies necessary properties.  We apply the phase space metric to several binary classification tasks, including discrimination of high-multiplicity resonance decays or boosted hadronic decays of electroweak bosons from QCD processes, and demonstrate powerful performance on simulated data. Our work demonstrates the many benefits of promoting phase space from merely a background on which calculations take place to being geometrically entwined with a theory’s dynamics.}

\maketitle
\flushbottom

\newpage

\section{Introduction}
\label{sec:intro}

Novel representations of collider data play an increasingly important role in the modern machine learning (ML) era \cite{Larkoski:2017jix,Guest:2018yhq,Albertsson:2018maf,Radovic:2018dip,Carleo:2019ptp,Bourilkov:2019yoi,Schwartz:2021ftp,Karagiorgi:2021ngt,Shanahan:2022ifi}. In particular, a physically-motivated notion of the distance between collider events is a powerful input to ML-based methods for event classification and anomaly detection. Recently, optimal transport distances between energy flow distributions \cite{Komiske:2019fks, Komiske_2020, PhysRevD.105.076003,Cai:2020vzx, cai2021linearizedhellingerkantorovichdistance} or spectral functions \cite{Larkoski:2023qnv} have given rise to a new geometry of collider events, with widespread applications \cite{Komiske:2019jim,Cesarotti:2020hwb,Cesarotti:2020ngq,Stein:2020rou,Kasieczka:2021xcg,Howard:2021pos,DiGuglielmo:2021ide,Kansal:2021cqp,Collins:2021pld,Cheng:2020dal,CrispimRomao:2020ejk,ATLAS:2023mny,Craig:2024rlv, Gouskos_2023}. However, there is another natural distance between collider events that has yet to be fully explored: the one furnished by the phase space manifold on which scattering amplitudes are defined. 

Traditionally, little attention has been paid to the intrinsic geometry of the phase space manifold itself, except within the context of the structure of S-matrix elements (especially at the boundaries of the phase space). The deeper potential of the phase space manifold was recently highlighted in Ref.~\cite{Larkoski:2020thc} (as well as Refs.~\cite{Cox:2018wce,Henning:2019mcv,Henning:2019enq}), which presented a covariant description of the massless $N$-body phase space manifold $\Pi_N$ and demonstrated that it is isomorphic to the product space of an $N-1$-dimensional simplex $\Delta_{N-1}$ and a $2N-3$-dimensional hypersphere $S^{2N-3}$, $\Pi_N\cong\Delta_{N-1}\times S^{2N-3}$. Explicit global coordinates can then be constructed on the phase space, promoting it to a Riemannian metric space. This offers a new way to organize collider data that is firmly grounded in the underlying physics and furnishes a novel representation of collider data for any downstream machine learning application.

In this paper we develop and demonstrate a practical metric for collider events using the global phase space coordinates presented in Ref.~\cite{Larkoski:2020thc}, tailored in part to the structure of events at hadron colliders. Constructing a true metric from these phase space coordinates poses a number of challenges, which we address in turn. Given that the dimension of the phase space manifold depends on the particle multiplicity, exclusive clustering of events into a fixed number of $N$ ``particles'' enables the calculation of distances between different events on a phase space manifold of uniform dimension $3N-4$.\footnote{A common practice in analyzing variable particle number events in collider physics is to fix the multiplicity to be some very large number, expected to be larger than that in any given event.  Then, events are zero-padded with momentum four-vectors to ensure that every event has the same number of particles.  Here, we explicitly study the resolved multiplicity dependence of the phase space metric as a tool for discrimination, but this more inclusive approach is simpler in practice as it does not require knowledge of the preferred or natural dimensionality of phase space.} Ensuring that the phase space metric satisfies the identity of indiscernibles requires minimizing distances over $N!$ permutations of particles for each pair of events, which is an exponentially computationally intensive task. Instead, we efficiently approximate this exact minimization by first minimizing distances on the simplex part of the phase space manifold and then azimuthally rotating events to minimize distances on the hypersphere part. Together, these steps ensure that the metric on phase space satisfies the identity of indiscernibles and defines a physically principled prescription for event pre-processing. We further tailor the metric to global phase space coordinates relevant to hadron colliders, including both the center-of-mass frame for full events and the collinear or infinite momentum frame for individual boosted jets. 

Different types of collider events exhibit distinctive structures in phase space, making the phase space distance between events a potentially powerful tool for event classification and anomaly detection. We demonstrate this potential with a number of examples in which the phase space distances among ensembles of signal and background events are inputs to simple ML algorithms for classification. Distances on the phase space manifold lead to competitive classification both at the level of entire collider events and within the substructure of boosted jets. Taken together, our results suggest that distances on a suitably-defined phase space manifold provide a practical and effective representation of collider data with numerous applications.

This paper is organized as follows: In Sec.~\ref{sec:PSGeometry} we review the prescription for constructing the phase space manifold and explicit global coordinates in the center-of-mass frame presented in Ref.~\cite{Larkoski:2020thc} and develop an analogous prescription in the collinear limit for boosted objects. We then address the practical aspects of constructing a phase space metric for collider events in Sec.~\ref{sec:Implementation}, including, among others, the representation of different events in terms of a fixed number of particles and the efficient minimization of distances over particle permutations. We demonstrate the classification power of the phase space metric coupled to simple ML algorithms in Sec.~\ref{sec:PS4Collider} with an event-level Standard Model $t \bar t$ vs.~QCD dijet study, an event-level beyond-the-Standard Model Hidden Valley isotropic resonance decay vs.~QCD dijet study, and a boosted, hadronically-decaying $W$ boson vs.~QCD jet study. We summarize our conclusions and sketch future directions in Sec.~\ref{sec:conclusions}, reserving technical details about the construction of coordinates and the minimization over particle permutations, as well as the full result tables, for the three appendices.

\section{Intrinsic Geometry of the $N$-body Phase Space}
\label{sec:PSGeometry}

In this section, we review the formalism of Ref.~\cite{Larkoski:2020thc} for establishing the phase space manifold in the center-of-mass (CM) frame.  We restrict our study to assume that experimentally detected particles are all massless, which is a good approximation in the high energy collider environment of the LHC.  For our purposes in this paper, we will need more explicit results for application on realistic collider event data.  We will thus provide all necessary mappings from momentum space coordinates to the natural coordinates on the simplex-hypersphere formulation of the phase space and the corresponding Riemannian metric.  Additionally, we will construct the manifold on the phase space of a boosted jet, effectively in the infinite-momentum frame. Here the coordinate of the jet axis is special, breaking the event's O(3) rotational invariance in the CM frame. Further, the jet mass is not necessarily measured and constrained, which introduces a new coordinate.

\subsection{Phase Space in the CM Frame: A Product Manifold}
\label{subsec:ProductManifold}

We start with the volume form of the $N$-body phase space in terms of the four-momenta $p_k$ of the $N$ massless particles constituting an event in Cartesian coordinates, where
\begin{gather}
    \label{eq:volume_form_1}
    d\Pi_N=(2\pi)^{4-3N}\left[\prod_{k=1}^Nd^4p_k\,\delta^+(p_k^2)\right]\delta^{(4)}\left(Q-\sum_{k=1}^Np_k\right)\,.
\end{gather}
Here $Q=(Q,0,0,0)$ is the total energy of the event in the CM frame and $\delta^+(p_k^2)=\delta(p_k^2)\,\Theta(p_{0k})$ enforcing that all particles are on-shell with positive energy. Note that the phase space in Eq.~\eqref{eq:volume_form_1} has a mass dimension of $2N-4$, whereas its number of degrees of freedom is $3N-4$. Unless otherwise specified, in our paper the term ``dimension'' refers to the number of degrees of freedom of a system, as this determines the dimension of the phase space manifold. 

The on-shell constraints can be satisfied by re-expressing the four-momenta in terms of the two spinors for each particle, where
\begin{gather}
    \lambda_k^\alpha=\frac{1}{\sqrt{p_{0k}-p_{3k}}}
\begin{pmatrix}
    p_{0k}-p_{3k} \\ p_{1k}+ip_{2k}
\end{pmatrix}\,,\quad\quad
    \Tilde{\lambda}_k^{\Dot{\alpha}}=\frac{1}{\sqrt{p_{0k}-p_{3k}}}
\begin{pmatrix}
    p_{0k}-p_{3k} \\ p_{1k}-ip_{2k}
\end{pmatrix}\,,
\end{gather}
with $k=1,\dots,N$ labeling the particle and $\alpha,\dot\alpha=1,2$ labeling the spinor components. Note that we have $\Tilde{\lambda}_k^{\Dot{\alpha}} = \lambda_k^{\alpha*}$ because the momentum matrix $p_k \cdot \sigma = \lambda_k\tilde\lambda_k$ is Hermitian.

We now define two $N$-component vectors
\begin{align}
&\vec{u}\equiv\frac{1}{\sqrt{\sum_{j=1}^N(p_{0j}-p_{3j})}}(\lambda_1^1~\lambda_2^1~\dots~\lambda_N^1)\text{ , with components  }u_k=\frac{1}{\sqrt{Q}}\sqrt{p_{0k}-p_{3k}}\,,\\
&\vec{v}\equiv\frac{1}{\sqrt{\sum_{j=1}^N(p_{0j}+p_{3j})}}(\lambda_1^2~\lambda_2^2~\dots~\lambda_N^2)\text{ , with components }v_k=\frac{1}{\sqrt{Q}}\frac{p_{1k}+ip_{2k}}{\sqrt{p_{0k}-p_{3k}}}\,,
\end{align}
in the CM frame. Notice that $\Vec{u}$ is a real vector while $\Vec{v}$ is complex, and both are dimensionless in the unit of mass.  
Written in terms of these $\Vec{u}, \Vec{v}$ vectors, Eq.~\eqref{eq:volume_form_1} now becomes
\begin{gather}
\begin{split}
    \label{eq:volume_form_2}
    d\Pi_N
    =&(2\pi)^{4-3N}Q^{2N-4}\,\frac{d^Nu\,d^Nv}{\text{U}(1)^N}\,\delta(1-|\Vec{u}|^2)\,\delta(1-|\Vec{v}|^2)\,\delta^{(2)}(\Vec{u}^\dagger\Vec{v})\,,
\end{split}
\end{gather}
where the division by U(1) represents implicit restriction to one element of the little group action on the spinors. 
 
By introducing the new coordinate
\begin{align}\label{eq:rho}
    \rho_k\equiv u_k^2\,,
\end{align}
the $\vec u$ part of the volume form reduces to
\begin{align}\label{eq:volume_form_simplex_part}
    \frac{d^Nu}{\text{U}(1)^N}\,\delta(1-|\Vec{u}|^2)=\frac{1}{2^N}\prod_{k=1}^N [d\rho_k]\, \delta\left(1-\sum_{k=1}^N\rho_k\right)=\frac{1}{2^N}\,d\Delta_{N-1}\,,
\end{align}
where $d\Delta_{N-1}$ is the flat measure on a $N-1$-dimensional simplex. We therefore obtain the simplex part of the phase space manifold.

For the terms in Eq.~\eqref{eq:volume_form_2} involving $\Vec{v}$, we first enforce the orthogonality condition $\delta^{(2)}(\Vec{u}^\dagger\Vec{v})$ to reduce the dimension of $\Vec{v}$ by 1 via the removal of the last component of $\Vec{v}$, i.e., $v_N$. Thus, we have
\begin{align}\label{eq:v}
    d^Nv\,\delta(1-|\Vec{v}|^2)\,\delta^{(2)}(\Vec{u}^\dagger\Vec{v})=\frac{d^{N-1}v}{\rho_N}\,\delta(1-|\Vec{v}|^2-|v_N|^2)\,.
\end{align}
Note that the $\Vec{v}$ on the right hand side of Eq.~\eqref{eq:v} now has only $N-1$ components, and $v_N$ can be expressed in terms of the components of $\Vec{u}$ and $\Vec{v}$. We then perform a change of variables from $\Vec{v}$ to $\Vec{v}'$, where $\Vec{v}'$ is an $N-1$-component complex unit vector,
\begin{align}\label{eq:v_change_of_vars}
    \frac{d^{N-1}v}{\rho_N}\,\delta(1-|\Vec{v}|^2-|v_N|^2)=d^{N-1}v'\,\delta(1-|\Vec{v}'|^2)=dS^{2N-3}\,,
\end{align}
which is manifestly the measure for the sphere $S^{2N-3}$.

Putting it all together, we have shown that the phase space is a product manifold of the simplex $\Delta_{N-1}$ and the sphere $S^{2N-3}$:
\begin{align}\label{eq:ndepvol}
d\Pi_N = \frac{(2\pi)^4}{Q^4}\left(\frac{Q^2}{16\pi^3}\right)^N\, d\Delta_{N-1}\, dS^{2N-3}\,.
\end{align}
One can easily check that dimensionality of the degrees of freedom works since the sum of the simplex and sphere dimensions gives us $(N-1) + (2N-3) = 3N-4$, which is the correct phase space dimension for a system of $N$ massless particles.  Also, because we will always compare events with the same total energy $Q$, we will often rescale energies so that $Q = 1$, which eliminates the overall dimensionality of the phase space metric.

\subsection{Phase Space in the CM Frame: Metric and Explicit Coordination}
\label{subsec:Metric}

From this factorization of the phase space manifold into a simplex and a sphere, we can then introduce explicit, global coordinates.  First, let's explicitly work out the full transformation from $\Vec{v}$ to $\Vec{v}'$. This can be expressed by
\begin{equation}
\begin{gathered}
    \label{eq:v'}
\begin{pmatrix}
    v_1' \\ v_2' \\ \vdots \\ v_{N-1}'
\end{pmatrix}
    =\frac{1}{u_N(1-u_N^2)}\,\Hat{R}
\begin{pmatrix}
    v_1 \\ v_2 \\ \vdots \\ v_{N-1},
\end{pmatrix}
\end{gathered}
\end{equation}
where the matrix $\hat R$ is
\begin{align}\label{eq:spheremat}
    &\Hat{R}\\
    &\hspace{0.25cm}=
\begin{pmatrix}
    u_1^2 + u_N (1-u_1^2-u_N^2) & u_1 u_2 (1-u_N) & \cdots & u_1 u_{N-1} (1-u_N) \\
    u_1 u_2 (1-u_N) & u_2^2 + u_N (1-u_2^2-u_N^2) & \cdots & u_2 u_{N-1} (1-u_N) \\
    \vdots & \vdots & \ddots & \vdots \\
    u_1 u_{N-1} (1-u_N) & u_2 u_{N-1} (1-u_N) & \cdots & u_{N-1}^2 + u_N (1-u_{N-1}^2-u_N^2)
\end{pmatrix}\,.\nonumber
\end{align}
Note that the above matrix $\Hat{R}$ is a real, symmetric and positive-definite, and the reality of $\Hat{R}$ ensures that the real and imaginary components of $\vec{v}$ are not mixed with each other. Furthermore, the determinant of the matrix is $u_N^{-1}$, and inverting $\Hat{R}$ gives a Jacobian $J=u_N^2$, which ensures that $d^{N-1} v = u_N^2\, d^{N-1} v'=\rho_N\, d^{N-1} v'$.   The matrix of Eq.~\eqref{eq:spheremat} was not explicitly constructed in Ref.~\cite{Larkoski:2020thc}, so we present its derivation in App.~\ref{app:spherematrix}.

The vector $\Vec{v}'$ can now be conveniently expressed by a generalized spherical coordinate system,
\begin{equation}
\begin{gathered}
    \label{eq:xi_eta}
    v'_1=e^{-i\xi_1}\cos\eta_1\\
    v'_2=e^{-i\xi_2}\sin\eta_1\cos\eta_2\\
    \vdots\\
    v'_{N-2}=e^{-i\xi_{N-2}}\sin\eta_1\cdots\sin\eta_{N-3}\cos\eta_{N-2}\\
    v'_{N-1}=e^{-i\xi_{N-1}}\sin\eta_1\cdots\sin\eta_{N-3}\sin\eta_{N-2},
\end{gathered}
\end{equation}
where $\xi_k\in[0,2\pi]$ and $\eta_k\in[0,\pi/2]$. These coordinates explicitly enforce the normalization $|\Vec{v}'|^2=1$. In the case of $N=2$, the sphere part of the manifold is $S^{2\cdot 2-3}=S^1$, and so no $\eta$ coordinate remains, and we're left with only one parameter for the sphere part, i.e., $\xi_1$.

To summarize, we have shown that the $(3N-4)$ degrees of freedom of the phase space manifold are composed of two groups of coordinates, with $(N-1)$ $\rho$ coordinates on  the simplex $\Delta^{N-1}$, and $(N-1)$ $\xi$ together with $(N-2)$ $\eta$ coordinates on the hypersphere $S^{2N-3}$.

Armed with a system of explicit coordinates, we can now easily define the corresponding metric on the phase space manifold. The line element of the simplex is 
\begin{gather}
    ds^2_{\Delta_{N-1}}=\sum_{k=1}^{N}d\rho_k^2.
\end{gather}
The summation above is over $N$ coordinates despite the fact that the simplex has dimension $N-1$. Indeed, these coordinates are not all independent but are instead constrained via the $\delta$-function in Eq.~\eqref{eq:volume_form_simplex_part}. With a suitable change of variables, the simplex line element could be written using only $N-1$ coordinates, but this is not necessary for our purposes. The line element of the hypersphere follows a recursive description
\begin{gather}
    ds^2_{S^{2N-3}}=d\eta_{N-2}^2+\cos^2\eta_{N-2} \; d\xi_{N-1}^2+\sin^2\eta_{N-2} \; ds^2_{S^{2N-5}}\,,
\end{gather}
starting from $ds^2_{S^1}=d\xi_1^2$. 

In this paper, we are more interested in the distances between pairs of events assigned according to their positions on the phase space manifold, so these line elements can be promoted to honest Riemannian metrics that satisfy symmetry, identity of indiscernibles, and the triangle inequality.  For two events that we label $A$ and $B$, their metric distance on the simplex can be defined as
\begin{align}\label{eq:d_delta}
d_\Delta(\Vec{\rho}_A,\Vec{\rho}_B)\equiv\sqrt{(\Vec{\rho}_A-\Vec{\rho}_B)^2}=\sqrt{\sum_{k=1}^N(\rho_{kA}-\rho_{kB})^2}\,.
\end{align}
Distances on the hypersphere are determined by the angle between the two events, and therefore a simple choice of the metric on the hypersphere is
\begin{align}\label{eq:d_S}
d_S(\Vec{v}_A',\Vec{v}_B')\equiv\cos^{-1}(\Re{\Vec{v}_A'^\dagger \Vec{v}_B'})=\cos^{-1}\left(\Re{\sum_{k=1}^{N-1}v_{kA}'^* v_{kB}'}\right)\,,
\end{align}
where $\Re{\Vec{v}_A'^\dagger \Vec{v}_B'}$ is the real part of the inner product of the vectors $\vec v'_A$ and $\vec v'_B$. 

Any convex linear combination of these two metrics is itself a metric, and could therefore be a metric on the phase space manifold. However, the phase space volume element is fixed, which constrains the determinant of the metric. By demanding that the metric on the phase space reproduces the dependence on multiplicity $N$ in the volume element of Eq.~\eqref{eq:ndepvol}, we can express the phase space metric between two events $A$ and $B$ as dependent on a single parameter $c>0$, where
\begin{align}\label{eq:d_Pi}
d_\Pi(\Vec{\rho}_A,\Vec{v}_A';\Vec{\rho}_B,\Vec{v}_B')=\sqrt{\frac{1}{16\pi^2} \left( \frac{c}{4} \right)^{\frac{3-2N}{3N-4}}d^2_\Delta(\Vec{\rho}_A,\Vec{\rho}_B)+\frac{1}{4\pi^2} \left( \frac{c}{4} \right)^{\frac{N-1}{3N-4}}d^2_S(\Vec{v}_A',\Vec{v}_B')}\,.
\end{align}
Varying the value of $c$ then reweights the simplex and sphere contributions to the distance between events.  With only the volume of the phase space as a constraint, there is no natural value for $c$ and all choices give rise to equally valid phase space distances. 

Indeed, to ensure that $d_\Pi$ is a metric on the phase space, we must further enforce permutation invariance of the particles in the events.  Specifically, if the events were identical, $A = B$, then for two distinct orderings of the particles, the metric as defined now would in general be non-zero.  This would immediately fail the identity of indiscernibles requirement, and so we must additionally enforce that the metric is minimized over all permutations of particles in the events.  In practice, global minimization over $N!$ permutations is exponentially computationally intensive and unfeasible even for reasonably small $N$. In Sec.~\ref{sec:Implementation}, we will discuss simple preprocessing steps based on physical considerations that can be implemented to approximate the global permutation minimum, and, at the very least, ensure that the distance from an event to itself is 0.

In principle, one can further generalize the above definitions of the metrics on the simplex, the hypersphere, and their resulting product space, as long as the phase space volume element is preserved. This would give us freedom to modify the relative distance between events in different regions of the phase space. For example, one could define the distance such that more emphasis is put on the edge of the phase space in comparison to its bulk. Such generalized metrics may potentially lead to improved discriminative power than the ones defined above. However, we do not pursue this study here, in order not to unduly complicate our analysis. For us, our metrics on the simplex, sphere, and the phase space are simply defined by Eq.~\eqref{eq:d_delta}, Eq.~\eqref{eq:d_S}, and Eq.~\eqref{eq:d_Pi}, respectively.

\subsection{Phase Space in the Collinear Limit} \label{subsec:CollinearPS}

The phase space for an isolated jet---as a collimated, high-energy collection of particles---has its many distinct features. Our starting approximation for a jet is in the infinite momentum limit, in which momentum transverse to its axis is parametrically smaller than the momentum along the axis \cite{Fubini:1964boa,Weinberg:1966jm,Susskind:1967rg,Bardakci:1968zqb,Chang:1968bh,Kogut:1969xa}.  In this limit, there is no Lorentz transformation that can be performed to boost to the rest frame of the jet, and so it requires a new analysis of the corresponding manifold on which particle momenta live.

Without loss of generality, we assume that the net momentum of the jet lies along the $+\hat z$ axis and introduce the lightcone momentum components $p_0+p_3$ and $p_0-p_3$. For our purposes here, we define a jet as a collection of particles for which their net $p_0+p_3$ component is parametrically larger than the $p_0-p_3$ component, i.e., $p_0+p_3 \gg p_0-p_3$.  To find a jet, we require some net $p_0+p_3$ but are inclusive to the value of $p_0-p_3$, without further restrictions on the structure of the jet.  Additionally, as the axis of the jet is defined to be the direction of net momentum, the net momentum transverse to the axis is 0.  We denote these components of momentum transverse to the jet axis as $\vec p_\perp$.

With this set-up, the differential phase space volume of a jet that consists of $N$ massless particles is then:
\begin{align}
    \label{eq:volume_form_jet_2}    
    d\Pi_N^{(J)}=&(2\pi)^{4-3N}\left[\prod_{k=1}^Nd^4p_k\,\delta^+(p_k^2)\right]d(p_0-p_3)\,\delta\left(p_0+p_3-\sum_{k=1}^N(p_{0k}+p_{3k})\right)\\
    &\hspace{1cm}\times\delta\left(p_0-p_3-\sum_{k=1}^N(p_{0k}-p_{3k})\right)\delta^{(2)}\left(\sum_{k=1}^N\Vec{p}_{\perp k}\right).\nonumber
\end{align}
As with phase space in the CM frame, we introduce vectors $\vec u$ and $\vec v$ to automatically incorporate the on-shell constraints, where now the individual components are:
\begin{align}
& u_k = \sqrt{\frac{p_{0k}-p_{3k}}{p_0-p_3}}\,, 
& v_k = \frac{p_{1k}+ i\, p_{2k}}{\sqrt{(p_0+p_3)(p_{0k}-p_{3k})}}\,.
\end{align}
Again, these vectors are dimensionless in the mass unit, which is preferred because to good approximation jets are scale invariant. Phase space volume in these coordinates then becomes
\begin{gather}
\begin{split}
    d\Pi_N^{(J)}=&(2\pi)^{4-3N}(p_0+p_3)^{2N-4}\left(\frac{p_0-p_3}{p_0+p_3}\right)^{N-2}\frac{d^Nu\,d^Nv}{\text{U}(1)^N}\,d(p_0-p_3)\\
    &\times\delta(1-|\Vec{u}|^2)\,\delta(1-|\Vec{v}|^2)\,\delta^{(2)}\left(\Vec{u}^\dagger\Vec{v}\right).
\end{split}
\end{gather}

Here we can go a bit further and simplify the phase space volume by incorporating some additional physics. In the soft and/or collinear limit that dominates the description of a jet, perturbative QCD is approximately scale-invariant, which imposes structure on the squared matrix element $|\mathcal{M}|^2$ as a function of the present scales. By dimensional analysis, the squared matrix element for $N$ final state particles must have overall dependence on $p_0+p_3$ as
\begin{align}
    |\mathcal{M}|^2 \propto (p_0+p_3)^{4-2N}\,,
\end{align}
in order to ensure that probabilities are dimensionless in the mass unit. Further, assuming scale invariance, probabilities are unchanged if the ratio
\begin{align}
\chi \equiv \frac{p_0-p_3}{p_0+p_3}\,,
\end{align}
is rescaled as $\chi \to \lambda \chi$, for any $\lambda > 0$, and so the matrix element must have dependence on $\chi$ as
\begin{align}
    |\mathcal{M}|^2 \sim \chi^{1-N} (p_0+p_3)^{4-2N}\,.
\end{align}
Including this knowledge of the matrix element, the phase space now becomes
\begin{align}\label{eq:jetpsvol}   
d\Pi_N^{(J)}\,|{\cal M}|^2 \sim (2\pi)^{4-3N}\,d\ln\chi\,\frac{d^N u\, d^N v}{\text{U}(1)^N}\, \delta\left(1 - |\vec u|^2\right)\delta\left(1 - |\vec v|^2\right)\delta^{(2)}\left(\vec u^\dagger \vec v
\right)\,.
\end{align}

If one is inclusive over particle number $N$, then this exact scale invariance in QCD is modified by an anomalous dimension (and the running coupling), proportional to $\alpha_s$. In general, the matrix element will exhibit singularities on the phase space when pairs of particles become collinear, but those details depend precisely on the specific dynamics of the jet ensemble of interest.

One can now proceed to construct explicit global coordinates on the jet's phase space manifold. As in the case of events in the CM frame, we will have coordinates $\rho,v',\xi,\eta$ expressed with the same formulae as in Eqs.~\eqref{eq:rho}, \eqref{eq:v'}, and \eqref{eq:xi_eta}. Therefore, we again end up with $(N-1)$ $\rho$ coordinates for the simplex and $(2N-3)$ $\xi$ and $\eta$ coordinates for the hypersphere. For a jet, we will also have an additional $\chi$ coordinate whose form is dictated by the specific structure of the squared matrix element. As a result, the phase space for collinear jets is a product of three parts: the simplex, the hypersphere, and a scale factor.

For two jets $A$ and $B$, the metric of the scale factor $\chi$ is simply the Euclidean metric on the real line:
\begin{align}
d_\chi(\chi_A,\chi_B)=\left|\log(\chi_A)-\log(\chi_B)\right|\,.
\end{align}
Now, with three components to the jet phase space manifold and one constraint on its volume form of Eq.~\eqref{eq:jetpsvol}, we can express the metric as a convex sum in terms of two parameters $c_1,c_2>0$, where 
\begin{align}\label{eq:jetPS_distance}
&d^2_{\Pi^{(J)}}(\Vec{\rho}_A,\Vec{v}_A',\chi_A;\Vec{\rho}_B,\Vec{v}_B',\chi_B)\\
&\hspace{0.5cm}=\frac{1}{16\pi^2} \left( \frac{c_1}{4} \right)^{\frac{3-2N}{3N-3}} \left( \frac{c_2}{16\pi^2} \right)^{-\frac{1}{3N-3}}d^2_\Delta(\Vec{\rho}_A,\Vec{\rho}_B)+\frac{1}{4\pi^2} \left( \frac{c_1}{4} \right)^{\frac{N}{3N-3}} \left( \frac{c_2}{16\pi^2} \right)^{-\frac{1}{3N-3}}d^2_S(\Vec{v}_A',\Vec{v}_B')\nonumber\\
&\hspace{1cm}+ \left( \frac{c_1}{4} \right)^{\frac{3-2N}{3N-3}} \left( \frac{c_2}{16\pi^2} \right)^{\frac{3N-4}{3N-3}}d_\chi^2(\chi_A,\chi_B),
\nonumber
\end{align}
where $d_\Delta$ is again given by Eq.~\eqref{eq:d_delta} and $d_S$ by Eq.~\eqref{eq:d_S}.  In some jet studies, there is an expected natural mass scale, for example, when searching for resonances that decay hadronically, and in those cases, one would impose a narrow cut on the jet mass.  As such, the value of $\chi$ would vary very little on a jet-by-jet basis in the relevant ensemble, and so the contribution to the metric from $d_\chi$ can typically be ignored, as is done in our later analysis. Additionally, same as events in the CM frame, permutation invariance of the particles must be imposed externally to ensure that the distance between a jet and itself is 0.

\section{Implementation of the Phase Space Metric for Classification}
\label{sec:Implementation}

In this section, we describe the practical implementation of the phase space metric that we will explore in specific case studies later. It involves significant pre-processing steps that go a long way toward realizing the permutation invariance and exact minimization required for the phase space distance to satisfy the requirements of a true metric.  We also describe the general techniques employed for post-processing the phase space representation of events and jets for classification, where machine learning is introduced to establish boundaries between signal- and background-rich regions. 

Pre- and post-processing of particle physics collider event data for any analysis, with machine learning or not, has long been employed to simplify the representation of the data, see, e.g., Refs.~\cite{Cogan:2014oua,Almeida:2015jua,deOliveira:2015xxd,Guest:2016iqz,Louppe:2017ipp,Cheng:2017rdo,Datta:2017rhs,Komiske:2017aww,Butter:2017cot}.  Typically, pre-processing is done in an {\it ad hoc} manner, modding out by symmetries that are expected to hold to good approximation, which is sometimes challenging to justify quantitatively. On the phase space manifold with explicit global coordinates and a local metric, these pre-processing steps are motivated by the necessary minimization over particle ordering and orientation to enforce the properties of the metric.  The steps discussed here can therefore inform a principled pre-processing prescription in other contexts.

There are necessarily many moving pieces in simplifying the event representation, evaluating the metric, and eventually classifying events. In the interest of clarity, we dedicate a subsection to each main processing step. Note that here we refer to both full collision events and individual jets excised from events as ``events'' and their clustered constituents as ``jets'', with the understanding that the appropriate description of the phase space is used depending on the particular application.

\subsection{Pre-processing of Collider Events}
\label{subsec:Preprocess}

We first discuss the pre-processing to perform on events. Such pre-processing steps ensure that the phase space distance between an event and itself vanishes.

\subsubsection{Event Representation} \label{subsubsec:Repre}

Our formulation of a metric on the phase space explicitly fixes the particle multiplicity $N$.  The first step is therefore to map any given event onto an $N$-body phase space.  We accomplish this by clustering the event (jet) into $N$ jets (subjets) with the exclusive $k_T$ clustering algorithm \cite{Catani:1991hj,Catani:1993hr,Ellis:1993tq}, where $N$ is a pre-specified fixed number. Further, the formalism requires each of the $N$ constituents to be massless. We therefore manually substitute the energy of each jet by the new value $E = |\vec p|$, and verify that the ratio $m /E$ for each resulting object is indeed almost exactly zero (allowing for machine precision). 

For small $N$, we expect that this masslessness restriction is {\it not} a good approximation, as each of the $N$ constituents may consist of relatively hard, relatively wide-angle particles that generate a relatively large mass for the constituent as a whole.  However, as multiplicity $N$ increases, it is expected that this masslessness restriction weakens as small clusters of hadrons are resolved that have high energy, but are collinear with one another.  This is justified {\it a posteriori} by plotting the distributions of the ratio $m/E$ for the jets as a function of $N$ (before manually setting $m$ to be 0), which we show later in Sec.~\ref{sec:PS4Collider} for the specific classification tasks.  In principle, mass can be incorporated in the phase space manifold and calculations in perturbative QCD always assume massless partons.  We leave a detailed analysis of the limitations of this massless approximation to future work.

\subsubsection{Particle Ordering} \label{subsubsec:Ord}

Events on which only energies and momenta are measured have no natural ordering of their constituent particles; so any function of them must be permutation invariant.  For the phase space metric, this permutation invariance should in principle be implemented as a minimization of the distance between two events over all possible permutations of particles, which ensures that the metric then satisfies both the triangle inequality and the identity of indiscernibles.  

However, for every pair of events with $N$ particles, we would have to scan over all possible $N!$ distinct orderings to identify the minimum.  Such an exhaustive search is clearly computationally impossible even for moderate $N$. While minimization over an exponentially-large number of permuations is ripe for machine learning, here we propose instead a much more pedestrian approach that is provably correct for the simplex submanifold and thus motivates the prescription on the entire phase space.

Recall that the metric on the simplex is the Euclidean metric in $N$ dimensions:
\begin{align}
d^2_\Delta(\Vec{\rho}_A,\Vec{\rho}_B)=\min_{\sigma\in S_N}\sum_{i=1}^N(\rho_{iA}-\rho_{\sigma(i)B})^2\,.
\end{align}
Of course, the simplex is $N-1$ dimensional by energy conservation, but we will work with this manifestly permutation invariant form.  Also, here we have explicitly included the minimization over the permutations of particles in event $B$, where $\sigma$ is an element of the symmetric group $S_N$.  Without loss of generality, we can assume that the particles in event $A$ are, say, ordered monotonically decreasing in $\rho$ value.  Then, the permutation of the particles of event $B$ that minimizes this distance is also when the particles in event $B$ are ordered in decreasing $\rho$ value.  That is, 
\begin{align}
d^2_\Delta(\Vec{\rho}_A,\Vec{\rho}_B)=\sum_{\substack{i=1\\ \rho_{iA}>\rho_{(i+1) A}\\ \rho_{iB}>\rho_{(i+1) B}}}^N(\rho_{iA}-\rho_{iB})^2\,.
\end{align}
We provide an elementary proof that this ordering does indeed render the distance a metric on the simplex in App.~\ref{app:euclidminorder}.

The story is more complicated for the hypersphere, since there is a rotation from $v$ to $v'$ coordinates. Of course, ultimately $d^2_\Pi$ is the correct distance to be minimized for each pair of events. Since $d^2_\Pi$ is a convex linear combination of the squared metrics on the simplex and the hypersphere, ordering particles according to their $\rho$ value does not guarantee minimization over permutations of the complete metric on the phase space.  

Nevertheless, this pre-processing step is very simple and computationally cheap, and at the very least approximates a local minimum over permutations.  To practically implement this ordering, recall that the definition of $\rho$ is
\begin{align}
\rho = p_0+p_3\,,
\end{align}
i.e., the sum of the energy and $z$-component of the particle momentum. For events in the CM frame, there is possibly an O(3) rotation symmetry of particle momenta.  For events at a hadron collider, this O(3) symmetry is broken to O(2) about the beams and reflection of the beams, but still, the direction of the $+z$-axis is not unambiguously defined.  In our analyses, then, we will consider the following three different particle orderings to probe the sensitivity to the choice of axis:
\begin{itemize}
\item Descending order in $ p_0+p_3$\,,
\item Descending order in $ p_0-p_3$\,,
\item Descending order in $p_T \equiv \sqrt{p_1^2+p_2^2}$\,.
\end{itemize}

Note that there is little a priori reason why one particle ordering is preferred for the full phase space distance. Our proposal here is inspired by the structure of the simplex alone, and we hope that the above physics-grounded orderings already give satisfying performance and further that the performance does not depend critically on the orderings. Future study of a better and more specialized way to search over the particle ordering space will hopefully improve the performance even more.

\subsubsection{Azimuthal Rotation} \label{subsubsec:Rot}

Aside from the particle ordering, there are still minimization ambiguities on the hypersphere that can be addressed with pre-processing.  Recall that the metric on the hypersphere is
\begin{align}
d_S(\Vec{v}_A',\Vec{v}_B')=\min_{\arg(\vec v_B')}\cos^{-1}\left(\Re\left[\vec v_A'^\dagger \vec v_B'\right]\right)\,.
\end{align}
Here the minimization is a bit subtle.  Event-by-event, there is an O(2) azimuthal symmetry about the beam, for which the momentum transverse to the beam can be rotated by a common angle or reflected.  In the metric, this SO(2) rotation invariance is manifest as minimization over the argument of the vector $\vec v_B'$ (we fix the vector $\vec v_A'$ because only their relative phase is important).  We will address reflection invariance shortly.

Because $\cos^{-1}$ monotonically decreases as its argument increases from 0, the minimization over the phase of $\vec v_B'$ is implemented by maximizing
\begin{align}
\max_{\arg(\vec v_B')}\Re\left[\vec v_A'^\dagger \vec v_B'\right] \leq \left|\vec v_A'^\dagger \vec v_B'\right|\,.
\end{align}
Because the real part of this vector dot product is bounded from above by its absolute value, even when maximizing over the argument of $\vec v_B'$, we should just use the absolute value in calculating the metric on the hypersphere.  This prescription will minimize the distance on the hypersphere for any given particle ordering, but does not ensure that the ordering itself renders a global minimum.  Nevertheless, with the ordering prescription on the simplex, this procedure will produce at least a local minimum of the metric on the entire phase space manifold.

To account for reflections in the transverse space, note that a reflection of $p_y \to -p_y$ corresponds to complex conjugation in the vector $\vec v$.  Therefore, to minimize the hypersphere metric over reflections with a given particle ordering, we simply choose the larger of $\left|\vec v_A'^\dagger \vec v_B'\right|$ and $\left|\vec v_A'^\intercal \vec v_B'\right|$.  Again, we just need to reflect (i.e., complex conjugate) vector $\vec v_A'$ in the product because the metric is only sensitive to the relative orientation between events.  

For a given particle ordering, the final form of the metric on the hypersphere is thus
\begin{align}
d_S(\Vec{v}_A',\Vec{v}_B')=\cos^{-1}\left(\max\left[\left|\vec v_A'^\dagger \vec v_B'\right|,\left|\vec v_A'^\intercal \vec v_B'\right|\right]\right)\,.
\end{align}
Crucially, the particle ordering on the simplex ensures that the distance between an event and itself vanishes, and the above prescription for minimization over azimuthal symmetries ensures that the distance on the hypersphere also vanishes. Taken together, these pre-processing prescriptions explicitly ensure that the metric on the phase space satisfies the quality of identity of indiscernibles.

\subsection{Processing the Phase Space Manifold}
\label{subsec:Process}

With the pre-processing steps defined above, we now turn to discuss practical evaluation of the coordinates on the phase space manifold and the corresponding distances between realistic events. Here we encounter the additional freedom in the definition of the simplex and sphere coordinates.

If events enjoy invariance to O(3) transformations of the celestial sphere and all particles are detected, then there is a continuous infinity of possible choices for the $\vec u$ and $\vec v$ coordinates, all of which are equally valid.  This is because, in the center-of-mass frame, the net momentum is 0 and so the net momentum transverse to any fixed axis is also 0; therefore any axis can legitimately serve as the $z$-axis, leading to a valid metric.  

However, here we either consider events produced at a hadron collider, for which numerous particles are lost down the beam region and undetected, or individual jets for which a center-of-mass frame does not exist.  These systems explicitly break the infinite degeneracy of coordinate systems and reduce the $\vec u$ and $\vec v$ coordinates to just a couple of choices that respect the symmetries of the events in the lab frame.

These possible choices are exclusively set by the definition of the $\vec u$ coordinates on the simplex, or equivalently, $\rho$.  Given fixed coordinates for which the $\hat z$-axis lies along the beams and one beam is fixed to be in the $+\hat z$ direction, there are two possibilities:
\begin{align}
\rho = p_0+p_3\,, \qquad \text{or}\qquad \rho = p_0-p_3\,.
\end{align}
Note that either choice produces a valid metric, because there is no absolute frame for the coordinates on the phase space. In principle, any valid definition of the phase space metric contains the same information content and therefore should be equally powerful for the downstream classification. Later we will study the response of a few classification tasks with both of the above choices of the $\vec u$ coordinates, denoted as ``\texttt{+ on simplex}'' and ``\texttt{- on simplex}'' accordingly.  

Note that this choice of coordinates ensures that the momentum transverse to the $z$-axis, which defines the coordinates $\vec v$, is (very nearly) 0, up to the small residual transverse momentum lost down the beams.  However, the visible particles in general have a large net momentum along the $z$-axis, which will need to be boosted away to render the events in the center-of-mass frame.  Such a boost does not affect the direction of the $z$-axis, and further the $\rho$ coordinate transforms homogeneously under a $z$-axis boost, i.e.,
\begin{align}
\rho \to e^{\eta} \rho\,,
\end{align}
where the (pseudo)rapidity of the visible particles in the lab frame is $\eta$.

Given a choice for the $\vec u$ coordinates, the $\vec v$ coordinates are then unique.  Conservation of transverse momentum requires that $\vec v$ is transverse to $\vec u$, i.e., $\vec u^\dagger \vec v = 0$; so $\vec v$ must be formed from a complex linear combination of $p_1$ and $p_2$.  We can in principle rotate the $x$ and $y$ axes about the $z$ axis however we desire, naively producing another continuous infinity of axis choices.  However, focused on ultimately evaluating the metric distance between events, we have already optimized this azimuthal rotation in the implementation of the metric on the hypersphere.  Global event rotations about the $z$ axis, as would be used to define new $(x,y)$ axes, cannot affect the metric distance. Therefore we can happily take whatever axes are provided to us in the representation of the data.

For studies on jets, we simply note that all of the above considerations apply there as well. The only modification is that we set the jet axis to be the $+\hat z$ direction, and the two choices of $\rho$ become the only legitimate ones, because the covariant phase space derivation relies on there being zero total momentum transverse to whatever momentum direction appears in the simplex coordinate; otherwise, we would have $\vec u^\dagger \vec v \neq 0$. We can then proceed to compute the phase space coordinates according to the prescriptions in Sec~\ref{sec:PSGeometry}.

\subsection{Post-processing for Classification}
\label{subsec:Postprocess}

After obtaining the simplex and the sphere distances, further post-processing steps are necessary in order to accomplish the task of supervised event classification. First, one needs to choose the relative weighting between the simplex and sphere parts (and $\chi$ if included) to define the total phase space metric. One then needs to pick a classifier which takes the phase space distance as the input. The two steps are indeed coupled together in the sense that a good choice of the relative weighting is determined solely by the performance of the classifier. In principle, different classifiers may prefer different values for the relative weighting, although one would expect a good classifier to deliver relatively stable performance across a range of reasonable values for the relative weighting.

\subsubsection{Relative Weighting between the Simplex and the Sphere} \label{subsubsec:Weight}

To define the phase space distance between any pair of events, we need to specify the free coefficient $c$ in Eq.~\eqref{eq:d_Pi} (or $c_1$ and $c_2$ in Eq.~\eqref{eq:jetPS_distance}). To simplify our analysis, here we do not consider variation of the weighting in front of $\chi$. Therefore, the only free parameter to tune is $c$ in Eq.~\eqref{eq:d_Pi}. 

Obviously, there are infinitely many choices and any finite value for $c$ gives rise to a valid metric. In a perfect world with infinite computational power, we would want to continuously vary the weighting to maximize tagging performance, quantified in a single number called Area under the ROC Curve (AUC), where an AUC of 1 means a perfect classifier and a value close to 0.5 signifies a random guess.  

Fortunately, varying the weighting does not entail redoing the actual distance matrix computation. Rather, we can just take the existing two distance matrices (one for the simplex and the other for the hypersphere) and create different linear combinations of them. In practice, we need to ensure relatively balanced contributions from each part to the total phase space distance. Therefore, in our study we only consider three reasonable values for the relative weighting, i.e., sphere : simplex = $1:4$, $1:1$, and $4:1$.

\subsubsection{Machine Learning Classifier Choice} \label{subsubsec:MLmodel}

In this study, we choose Support Vector Machine (SVM) \cite{cortes1995support} as the classifier for event and jet tagging, thanks to its simplicity and robustness.\footnote{We have also explored the performance of the $k$-nearest neighbors ($k$NN) algorithm, another simple and robust supervised learning classifier. We find that SVM consistently outperforms $k$NN in all classification tasks, which presumably reflects the natural clustering of events of different types on the phase space manifold.} In general, we expect any algorithm (including neural networks) that can take distance matrices as inputs to work well with the phase space formalism, and encourage further exploration of more advanced methods.

SVM separates data points into two classes by finding a hyperplane whose distances to the nearest data points are maximized on both sides. The generalized notion of distance is known as kernel. We choose the commonly-used Radial Basis Function (RBF) kernel \cite{powell1977restart,broomhead1988multivariable} $k(\Vec{x}_A;\Vec{x}_B)=\exp\{-\gamma d(\Vec{x}_A;\Vec{x}_B)^2\}$, where in our case $d(\Vec{x}_A;\Vec{x}_B)=d_\Pi(\Vec{\rho}_A,\Vec{v}'_A;\Vec{\rho}_B,\Vec{v}'_B)$ or $d(\Vec{x}_A;\Vec{x}_B)=d_\Pi(\Vec{\rho}_A,\Vec{v}'_A,\chi_A;\Vec{\rho}_B,\Vec{v}'_B,\chi_B)$ is the phase space distance. 

There are two hyperparameters $C$ and $\gamma$ in the SVM model. The kernel parameter $\gamma$ can be interpreted as the inverse of the radius of influence of training data points. On the other hand, the soft margin parameter $C$ controls the relative importance between correct classification and large separation, with a small $C$ valuing a large separation between the optimal hyperplane and nearest data points over correctly classifying all data points and vice versa.

For the event and jet classification tasks described in Sec.~\ref{sec:PS4Collider}, we generate 10k events or jets in each case, and further separate them into a training dataset of size 5000, a validation dataset of size 2500 for choosing the best model hyperparameters (i.e., $\gamma$ and $C$), and a test dataset of size 2500 to evaluate the model performance. We restrict the ranges of hyperparameters to be $C\in[10^{-2},10^3]$ and $\gamma\in[10^{-2},10^5]$ with a multiplication of 10 in each step. These values are chosen based on empirical observations with the goal to include the best combination of hyperparameters.

\section{Phase Space Manifold for Event Classification and Jet Tagging}
\label{sec:PS4Collider}

The separation of signal events from backgrounds is a key task for the search of new physics beyond the Standard Model (BSM) at the LHC. Oftentimes, the majority of background processes are made up of QCD events, whose dynamics in the perturbative regime is dominated by the emission of soft and collinear quarks and gluons due to the small 't Hooft coupling $\lambda = g_s^2 N_C$ of QCD at high collider scales. Much of the kinematic information from this initial perturbative showering is retained through subsequent hadronization, and the final state particles of QCD events, mostly populated by light mesons, display jet-like structure that occupy some certain marked region in the phase space manifold. 

In contrast, various kinds of signals (including both SM and BSM events) may show distinctively different phase space structures due to the unique nature of their respective underlying physics. For example, top-quark pair production, a typical SM process, gives rise to an event configuration of six prongs on average, inducing a more uniform distribution than the QCD background. This will in turn differentiate the phase space region occupied by $t \Bar{t}$ events from that of QCD events. 

One can further consider exotic BSM scenarios where high-multiplicity events are produced whose final state particles have uniform distribution on the $N$-body phase space. Indeed, many models can produce such high-multiplicity events. One general class is Hidden Valley models \cite{Strassler_2007}, which emerges from solutions to string theory constructions and the hierarchy problem. With large 't Hooft couplings in contrast to QCD, Hidden Valleys can yield uniform radiation patterns and is a good signature to compare against SM backgrounds such as QCD events. 

Here we demonstrate the utility of the phase space distance in three classification tasks: event-level classification of top pair production vs.~QCD dijet events in Sec.~\ref{subsec:tQCDEvents}; event-level classification of a Hidden Valley scenario with high-multiplicity uniformly-distributed events vs.~QCD dijet background events in Sec.~\ref{subsec:QCDRamboEvents}; and jet-level classification of  two-pronged boosted $W$ boson jets vs.~QCD background jets in Sec.~\ref{subsec:WQCDJets}. The event-level classification tasks are tailored in part for comparison with benchmarks for the event isotropy variable explored in Ref.~\cite{Cesarotti:2020hwb}.

For each classification task, we examine different definitions of the phase space distances, as detailed in Sec.~\ref{sec:Implementation}. Specifically, we try both definitions of the $\vec u$ coordinates, i.e., ``\texttt{- on simplex}'' and ``\texttt{+ on simplex}'', as well as three relative weightings between the simplex and the sphere parts, denoted as ``\texttt{sphere:simplex = 1:4}'', ``\texttt{sphere:simplex = 1:1}'', and ``\texttt{sphere:simplex = 4:1}''. This gives us 6 sub-tasks, all giving rise to equally valid phase space distances. One would expect the classification performance to vary minimally no matter which distance definition is used. This is because the exact same information, though combined in different ways, is encoded in any valid phase space distance and a good classifier should be able to utilize them equally well. Our intuition is largely supported by the three case studies. 

In contrast, within each sub-task, the choice of particle ordering and the number of jets (subjets) representing an event (jet), i.e., $N$, do make a material difference, with the former giving different approximations to the exact phase space distance and the latter defining distinct phase space manifolds themselves. The particle ordering and (sub)jet number that maximize tagging performance vary from case to case depending on the underlying collisions. In general, the masslessness assumption for the individual components in an event holds to better accuracy as $N$ grows higher. One therefore expects the current phase space formulation to give a more precise description of the actual events consisting of massive particles under the large $N$-representations, potentially leading to better downstream classification performance with increasing $N$.  

We perform all six sub-tasks for each of the three classification tasks, amounting to $\textbf{2} \text{ ($\vec u$ coordinates definitions)} \times \textbf{3} \text{ (particle orderings)} \times \textbf{8} \text{ (event representations)} = \textbf{48}$ rounds of phase space distance calculations. The relative weighting does not require re-calculation of the simplex and the sphere distances, as they can be combined according to the specific ratio to directly generate the total phase space distance. The calculation of the simplex and sphere distances for the entire dataset of 10k events takes a couple of hours on a single desktop (Intel Core i7-8700 CPU) for $N\leq5$, with computing time increasing sharply for larger $N>10$, e.g. taking about 50 hours for $N=40$. However, we expect optimization of the codes to significantly improve efficiency, and the use of cluster will further decrease computational time given that our method is highly amenable to parallel computing.

Performance is always quoted as the best AUC on the test set with the optimal hyperparameters chosen by the validation set. All results are shown as plots of AUC vs.~$N$-representation, with the actual AUC scores given in tables in App.~\ref{app:tables}.

\subsection{Top Pair Production vs.~QCD Dijet Events}
\label{subsec:tQCDEvents}

Proton-proton collision events at $\sqrt{s} = 14$ TeV are generated in \textsc{madgraph5} 2.9.6 \cite{Alwall_2014}, with top quarks being pair produced via $pp \to t \Bar{t}$, and QCD dijet events generated through $pp\to jj$ where $j$ represents a quark ($u, d, c, s$) or gluon. The top (or anti-top) quarks subsequently decay via $t \to W^+ q$ (or $\Bar{t} \to W^- \Bar{q}$) where $q = b, s, d$ and $\Bar{q}$ represents the corresponding anti-quarks. The $W$ bosons further decay hadronically, i.e., $W\rightarrow jj$, so that the top events in principle have six jets if all $j$'s are resolved. Particles are then hadronized and further decayed in \textsc{Pythia} 8.303 \cite{Sj_strand_2015}, where default tuning and showering parameters are used. Only final state particles with pseudo-rapidity $|y|<2$ in the lab frame are kept. We select events whose scalar transverse momentum is $\sum p_T>400$ GeV and cluster them into jets using \textsc{FastJet} 3.4.0 \cite{Cacciari_2012}. In total, we have 10k simulated events, with about half being signal.

To construct the phase space manifold, we first cluster each event into a certain number of jets using the $k_T$ exclusive algorithm. For each event, only the $N = 3, 4, 5, 10, 15, 20, 25, 30$ hardest jets are retained, where a jet is treated as a pseudo-particle with its kinematic information given by the jet axis and its mass manually set to zero. We then center each event by boosting it to its CM frame, and normalize the $(E, p_x, p_y, p_z)$ of each constituent jet by the total energy of the event when computing the phase space coordinates.

Fig.~\ref{fig:10ktQCDdijetEvents} visualizes a top pair production event and a QCD dijet event randomly selected from our 10k dataset. The events with all the constituent particles are shown first, followed by the various $N$-representations on the $y-\phi$ plane with $N = 3, 4, 5, 10, 15, 20, 25, 30$. Next, to demonstrate the appropriateness of the masslessness approximation, Fig.~\ref{fig:mptratio_tQCD} shows the distributions of each jet's $m/p_T$ on a log-log scale, before jet mass is manually set to zero. Clearly, most jets have negligible mass comparing to their transverse momenta, validating the massless phase space approximation especially at large $N$.

\begin{figure}[t!]
    \centering
    \includegraphics[width=\textwidth]{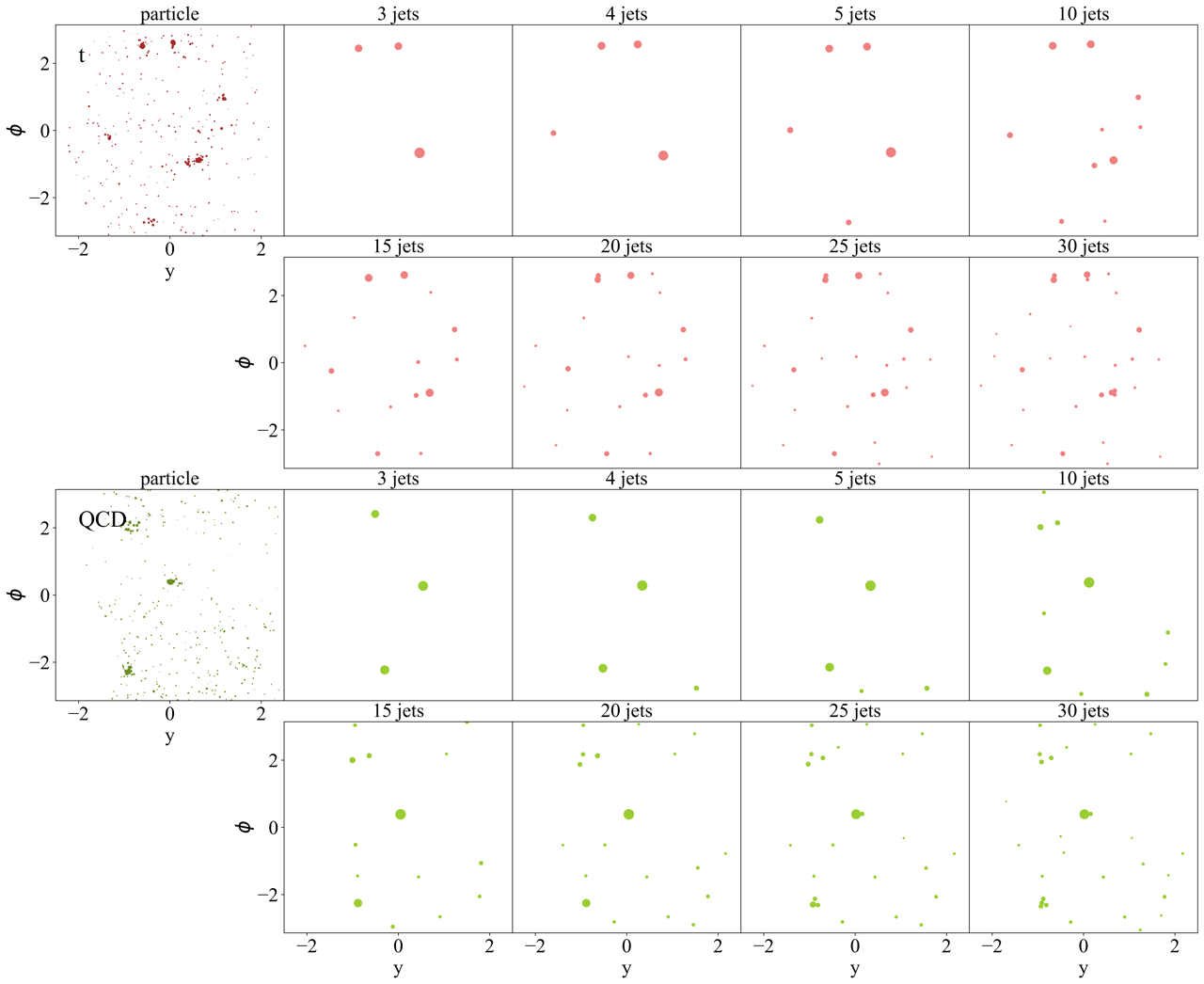}
    \caption{A $t \Bar{t}$ event (red) and QCD dijet event (green), represented by all the final state particles (leftmost plot) or $N$ exclusive jets ($N=3, 4, 5, 10, 15, 20, 25, 30$). The size of each dot is proportional to the $p_T$ of each constituent, plotted on the $y-\phi$ plane.}
    \label{fig:10ktQCDdijetEvents}
\end{figure}

\begin{figure}[t!]
    \centering
    \includegraphics[width=0.45\textwidth]{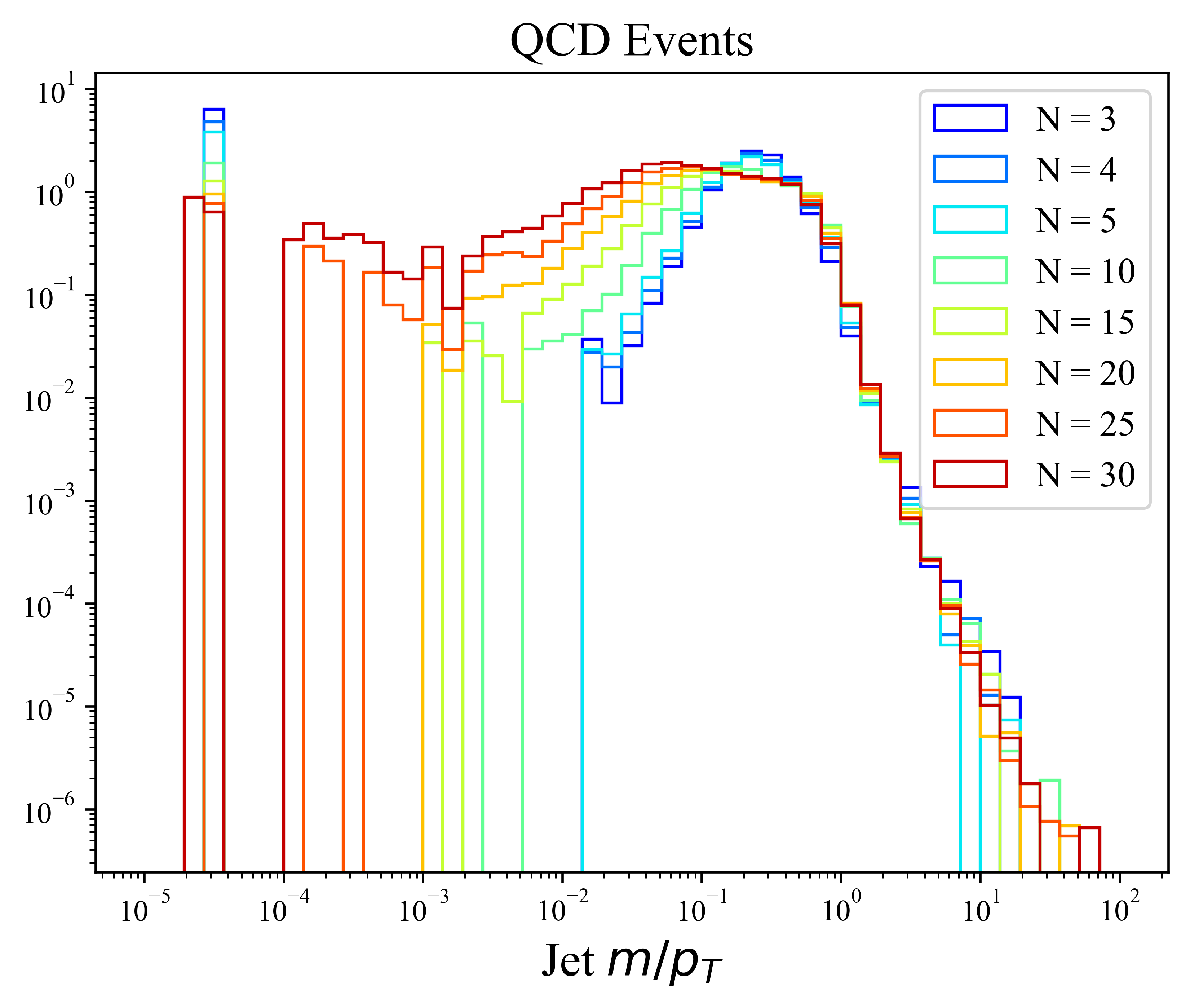}
    \includegraphics[width=0.45\textwidth]{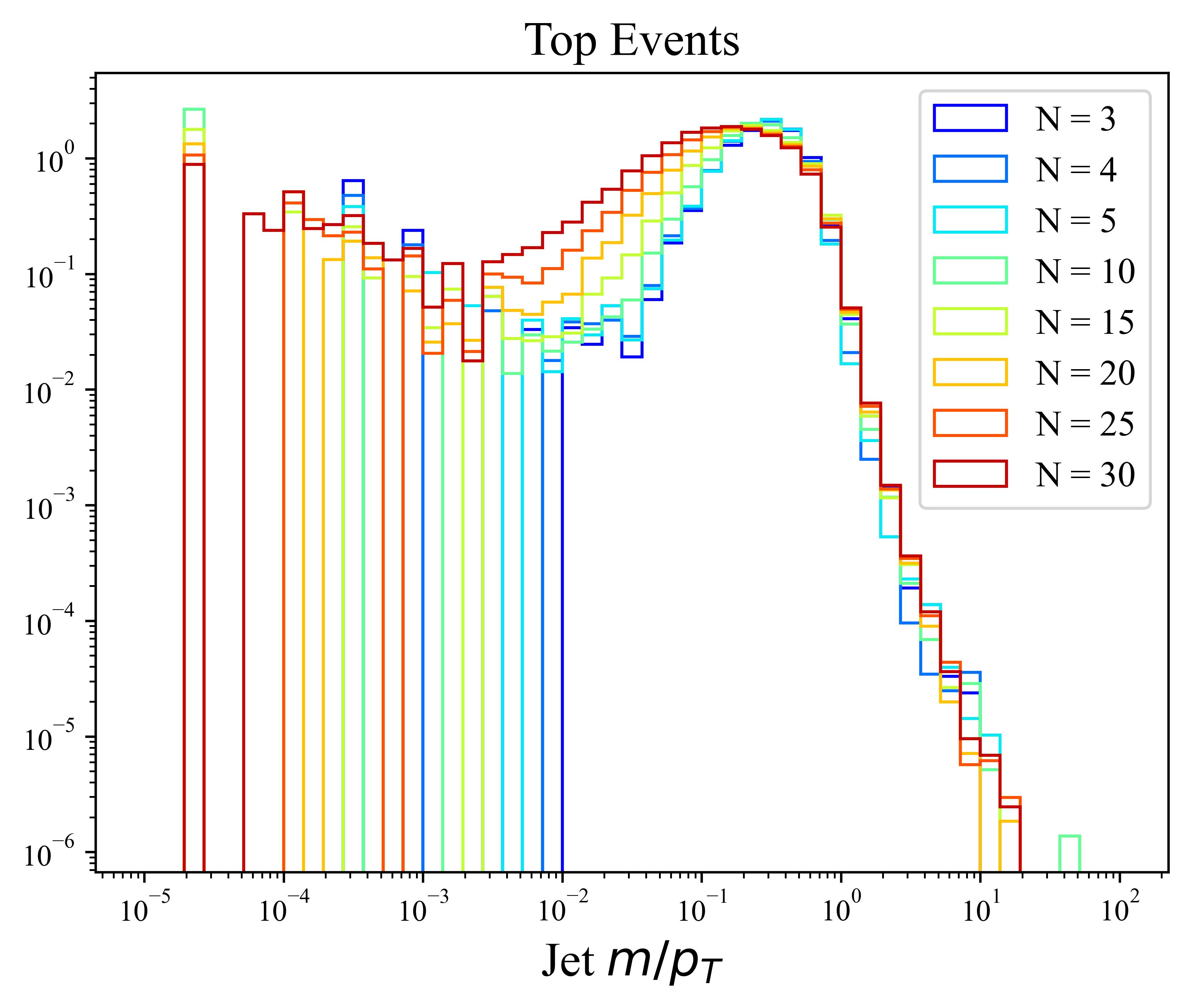}
    \caption{Histograms for each jet's $m/p_T$ (before manually setting $m$ to zero) of the QCD ({\it left}) and top ({\it right}) events. Note that both axes are on a log basis, and the colors denote different numbers $N$ of jets per event.}
    \label{fig:mptratio_tQCD}
\end{figure}

\subsubsection{Results}

We first show some sample histograms of the total phase space distances for our top and QCD events under a particular set of metric definitions (\texttt{- on simplex}, \texttt{sphere:simplex = 1:1}, $(p_0-p_3)$-ordering); see Fig.~\ref{fig:PStotal_tQCD}. Here we also separately plot the pairwise distances between two QCD events, two top events, and a QCD event vs.~a top event. Since the histograms aggregate all pairwise distances, the plots for the distances among different types of events look rather similar, even though the discriminative power of the total phase space distance for top vs.~QCD events is indeed very high. The difference of histograms for different types of events is, however, more noticeable when $N=4, 5$, indicating better classification performance around these $N$ values to be confirmed below.

\begin{figure}
    \centering
    \includegraphics[width=\textwidth]{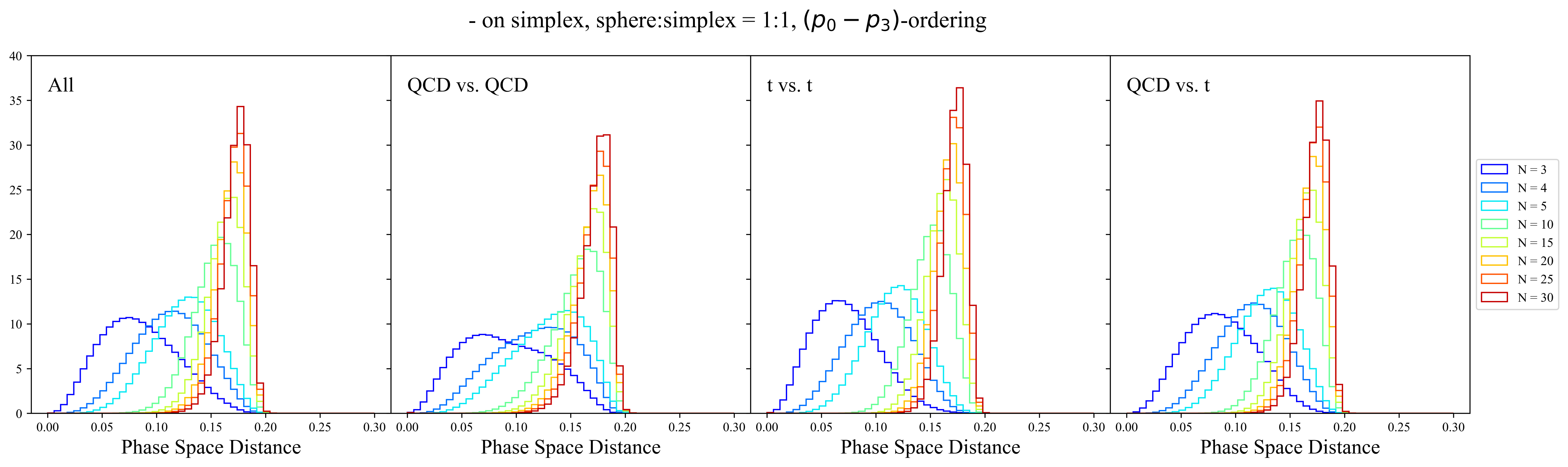}
    \caption{{\it From left to right:} Histograms of the total phase space distances between two arbitrary events in the dataset, between two QCD events, between two top events, and between a QCD event and a top event. The particular set of metric definitions is chosen to be \texttt{- on simplex}, \texttt{sphere:simplex = 1:1}, and $(p_0-p_3)$-ordering. The colors denote different numbers $N$ of jets per event.}
    \label{fig:PStotal_tQCD}
\end{figure}

Fig.~\ref{fig:QCDtevents_l_total_SVM} shows the classification performance of the total phase space distances coupled with SVM. The highest AUC obtained among all definitions is $\sim 0.86$, corresponding to the simplex definition of $\rho = p_0 - p_3$; a relative weighting of \texttt{sphere:simplex = 1:1}; the $(p_0+p_3)$-ordering; and an event representation of $N=4$ jets per event. This is quite remarkable, given that all we have utilized is the ``locations'' on the phase space manifold of different types of events. No explicit information about the underlying collision process is provided to the model. In comparison, the best performing event shape observable studied in Ref.~\cite{Cesarotti:2020hwb}, i.e., the ring-like event isotropy, achieves an AUC score of only 0.774, about $10\%$ lower than our top performer. Of course, caution must be taken to interpret this comparison, given that our datasets are not the same and that Ref.~\cite{Cesarotti:2020hwb} employed only a simple cut on the single observable whereas here we use SVM.

We make three overall observations for the performance of the phase space metric. First, as noted above, both the definitions of the simplex coordinates and the relative weighting between the simplex and the sphere make little difference on the final AUC score, suggesting that any valid phase space distance contains comparable information recoverable by SVM. Of course, this is only true within reasonable variation of the sphere-to-simplex weighting, as extreme weights asymptote to the sphere-only and simplex-only distances.  

Second, there is no obvious best particle ordering for approximating the exact phase space distance. All three orderings render roughly the same classification performance, especially in the low $N$ region ($N \le 10$) where the best AUC is achieved. As more jets are used to represent an event ($N > 10$), the overall performance degrades slightly and in most cases the descending $p_T$-ordering becomes the best choice among the three, giving stable performance across the entire range of $N$. But again, the difference is relatively small to draw statistically significant conclusions about the best particle ordering scheme. One would expect that an exact minimization of the phase space distance over all particle permutations would further improve its tagging power, although a more accurate (in this case smaller) distance does not necessarily entail a higher AUC, as will be shown momentarily.      

Instead, the largest impact on classification arises from event representation by $N$ (massless) jets, which defines the phase space manifold itself with a dimension of $3N-4$. The best performance is achieved when $N \sim 4, 5$ and the performance consistently deteriorates with increasing $N$. As exemplified in Fig.~\ref{fig:10ktQCDdijetEvents}, around 5 jets can capture most of the hard components inside a QCD dijet event or a $t\bar{t}$ event, which matches our theoretical understanding. Adding more objects means including soft particles which mainly originate from QCD for both signal and background, the inclusion of which may confound the downstream classifier. Another explanation could be that higher $N$ entails a phase space with higher dimensionality, potentially making it harder for SVM to classify the data. However, the other two classification tasks will soon show that this is not the case, as there the classification performance indeed increases with larger $N$. The $N$-dependence of classification appears to be more a feature of the underlying events being compared than the dimensionality of the phase space manifold itself.

\begin{figure}[t!]
	\centering
	\includegraphics[width=\textwidth]{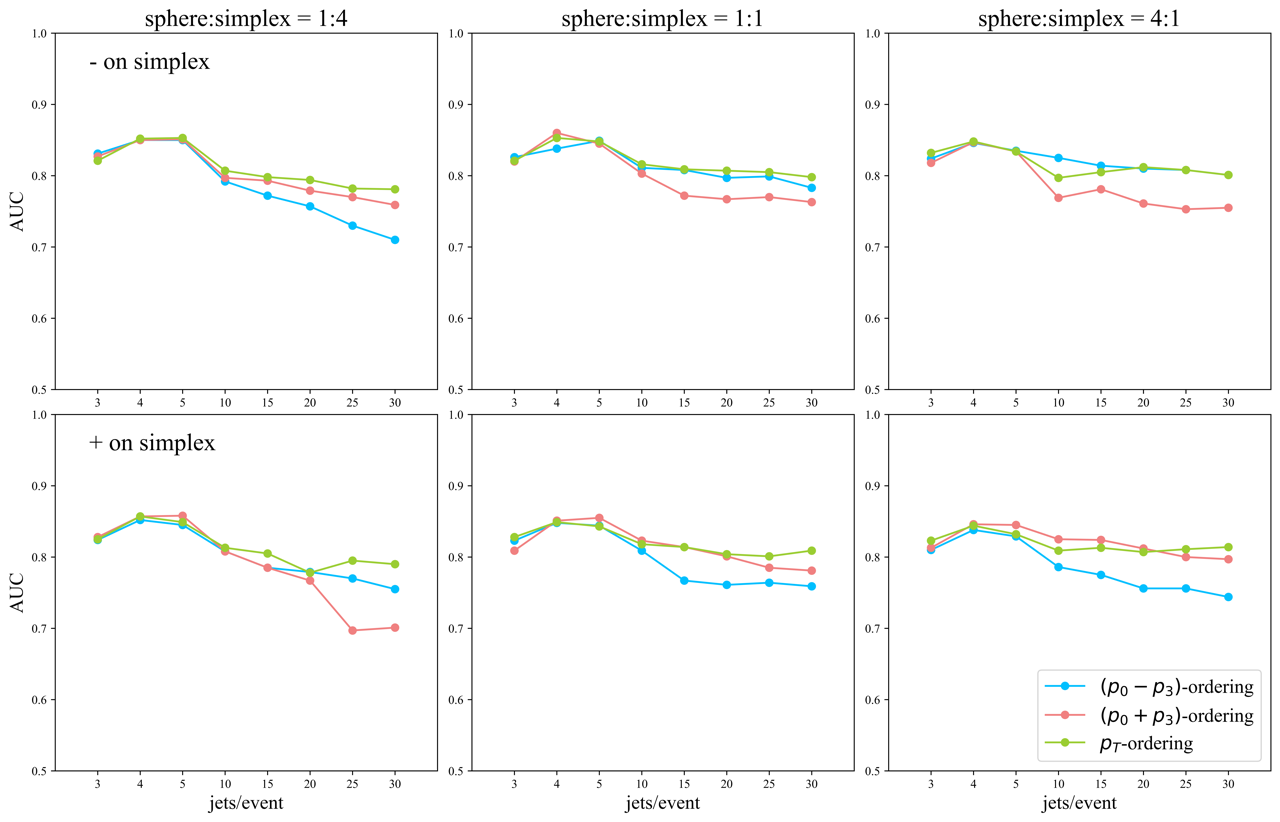}
	\caption{AUC values for QCD dijet events vs.~top pair production events, using SVM as the classifier. The number of jets $N$ to represent each event (i.e., the $x$-axis) is varied to be $N = 3, 4, 5, 10, 15, 20, 25, 30$. Each subplot corresponds to a distinct pair of $\vec u$ coordinates definition coupled with the sphere-to-simplex weighting ratio. The blue, red, and green curves represent results with constituents ordered by decreasing $(p_0-p_3)$, $(p_0+p_3)$, and $p_T$, respectively.}
	\label{fig:QCDtevents_l_total_SVM}
\end{figure}

We now zoom in to examine the simplex and sphere submanifolds separately under the two $\vec u$ coordinates definitions. Fig.~\ref{fig:QCDtevents_Sphere11_l_simplex&sphere_SVM} shows the AUC scores if the input distance matrix only consists of either the simplex or the sphere part of the phase space, which is equivalent to setting \texttt{sphere:simplex} as \texttt{0:1} or \texttt{1:0}, respectively. The goal is to disentangle the information content contained in the two submanifolds and see which part contributes most to the classification task.

For both simplex and sphere parts, the general trend of the AUC scores follows that of the total phase space, with performance again peaking around $N \sim 4, 5$ and gradually declining with larger $N$ by about $10 \sim 15\%$. Here the sphere part of the phase space contains more discrimination power than the simplex, where the best achievable AUC by the sphere distance is more than $20\%$ higher than that of the simplex distance under the same definitions. Indeed, the sphere distance alone gives almost comparable AUC as the total phase space distance, with a mere $\sim 3\%$ difference between their respective optimal scores. The major difference, however, is that the total phase space distance depends less on the specific particle ordering, whereas the difference among the three orderings is much more pronounced for the two submanifolds, especially on the simplex.   

Interestingly, given a particular definition of the $\vec u$ coordinates, the best particle ordering corresponding to the exact minimum distance on the simplex nonetheless performs the worst when only the simplex distance matrix is inputted to the model. For example, with $\rho$ defined to be $(p_0 + p_3)$ (i.e., \texttt{+ on simplex}), the particle ordering that minimizes the distance on the simplex is the $(p_0 + p_3)$-ordering, which consistently gives the lowest AUC scores across all values of $N$ for the simplex alone. In this case, it is always the $(p_0 - p_3)$-ordering that delivers the best performing simplex distance, quite contrary to our intuition. 

On the other hand, the particle ordering that minimizes the simplex distance gives rise to a sphere distance that performs the best for lower $N$ ($N \leq 15$), where there is no obvious reason why one ordering is favored more than another on the sphere. The intermediate $p_T$-ordering turns out to generate the most stable performance over all $N$, both for the simplex alone and for the sphere alone. Such behavior is common to all three classification tasks, and does not appear to be an artifact of the particular events under consideration.

Of course, analyzing the two submanifolds of the phase space separately may not be physically meaningful. Rather one should always use the total phase space distance, in which case all three particle orderings give rise to similar performance with a slight preference for the $p_T$-ordering, thanks to its stability across $N$. Yet one general lesson we can learn from the simplex-only analysis is that an exact minimum distance does not guarantee best performance for classification.

\begin{figure}[t!]
	\centering
	\includegraphics[width=0.8\textwidth]{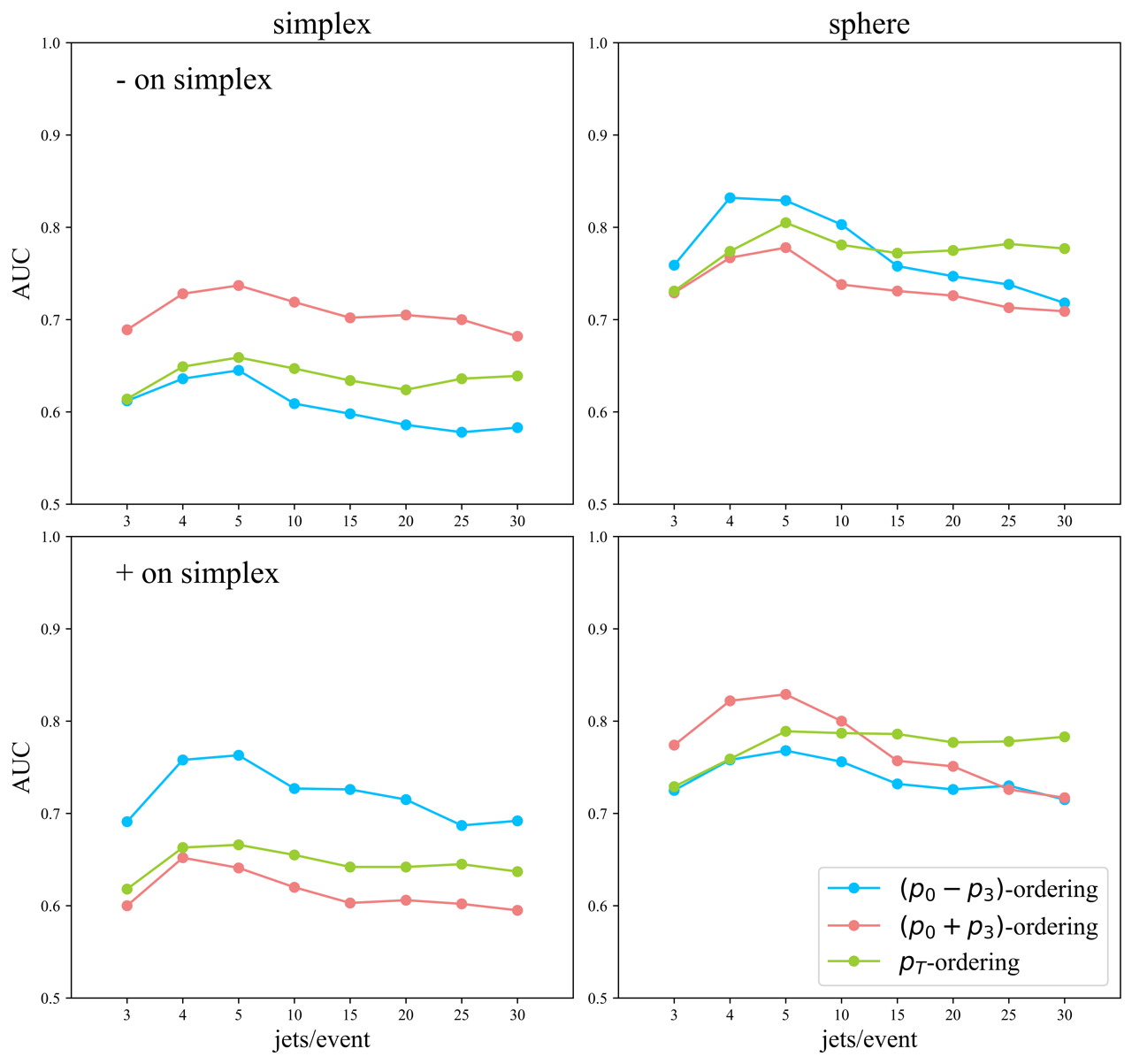}
	\caption{AUC values of QCD dijet events vs.~top pair production events, where the input distance matrix is the simplex ({\it left}) or the sphere ({\it right}) part alone under both $\vec u$ coordinates definitions ({\it upper} for \texttt{- on simplex} and {\it lower} for \texttt{+ on simplex}).}
	\label{fig:QCDtevents_Sphere11_l_simplex&sphere_SVM}
\end{figure}


\subsection{A New Physics Toy Model with Uniform $N$-body Phase Space Events}
\label{subsec:QCDRamboEvents}

We now consider as our signal a Hidden Valley scenario with a heavy mediator $Z'$ that generates high-multiplicity events uniformly distributed in $N$-body phase space. The backgrounds are still QCD dijet events, generated in the same way as before, except that the scalar transverse momentum cut is now raised to $\sum p_T>650$ GeV.  

The Hidden Valley model is generated via \texttt{HiddenValley:ffbar2Zv} in \textsc{Pythia} 8.303 with default tuning and showering parameters. The nominal mass of the $Z'$ is set to be 1 TeV with a width of 20 GeV, and a mass cut of  $m_{Z'} \in [750, 1250]$ GeV. Only final states with $|y|<2$ are kept. The particle $Z'$ is decayed using the \textsc{Rambo} algorithm \cite{KLEISS1986359}, where its $N_{Z'}$ decay products are distributed uniformly in the rest frame of the $Z'$. Here we pick $N_{Z'}=10$ to demonstrate the high-multiplicity of the new physics events. We then boost the decay products of $Z'$ back to the lab frame. After applying again a $|y|<2$ cut for these particles, we combine the remaining decay products with the rest of the \textsc{Pythia}-generated final-state particles in the same event. Finally, we apply a cut $\sum p_T>650$ GeV on the scalar transverse momentum of each event. 

Both signal and background events are clustered into $N$ jets in \textsc{FastJet} using the $k_T$ exclusive algorithm, with $N=5, 10, 15, 20, 25, 30, 35, 40$ hardest jets in each event to form various phase space representations. Fig.~\ref{fig:NewPhysicsEvents} visualizes on the $y-\phi$ plane a new physics event represented as $N$ jets, and Fig.~\ref{fig:mptratio_QCDZ'} again displays the distributions of $m/p_T$ for the jets on a log-log scale. The massless phase space approximation is clearly appropriate for the majority of the events, even at low $N$.

\begin{figure}[t!]
	\centering
	\includegraphics[width=\textwidth]{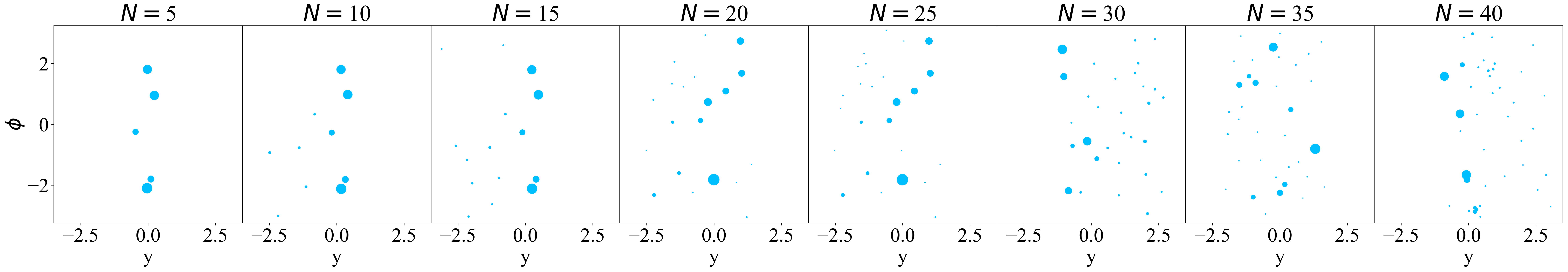}
	\caption{A new physics event on the $y-\phi$ plane clustered into $N = 5, 10, 15, 20, 25, 30, 35, 40$ jets (i.e., the number of dots in each plot), where the size of every dot is proportional to the $p_T$ of the constituent particle.}
	\label{fig:NewPhysicsEvents}
\end{figure}

\begin{figure}[t!]
    \centering
    \includegraphics[width=0.45\textwidth]{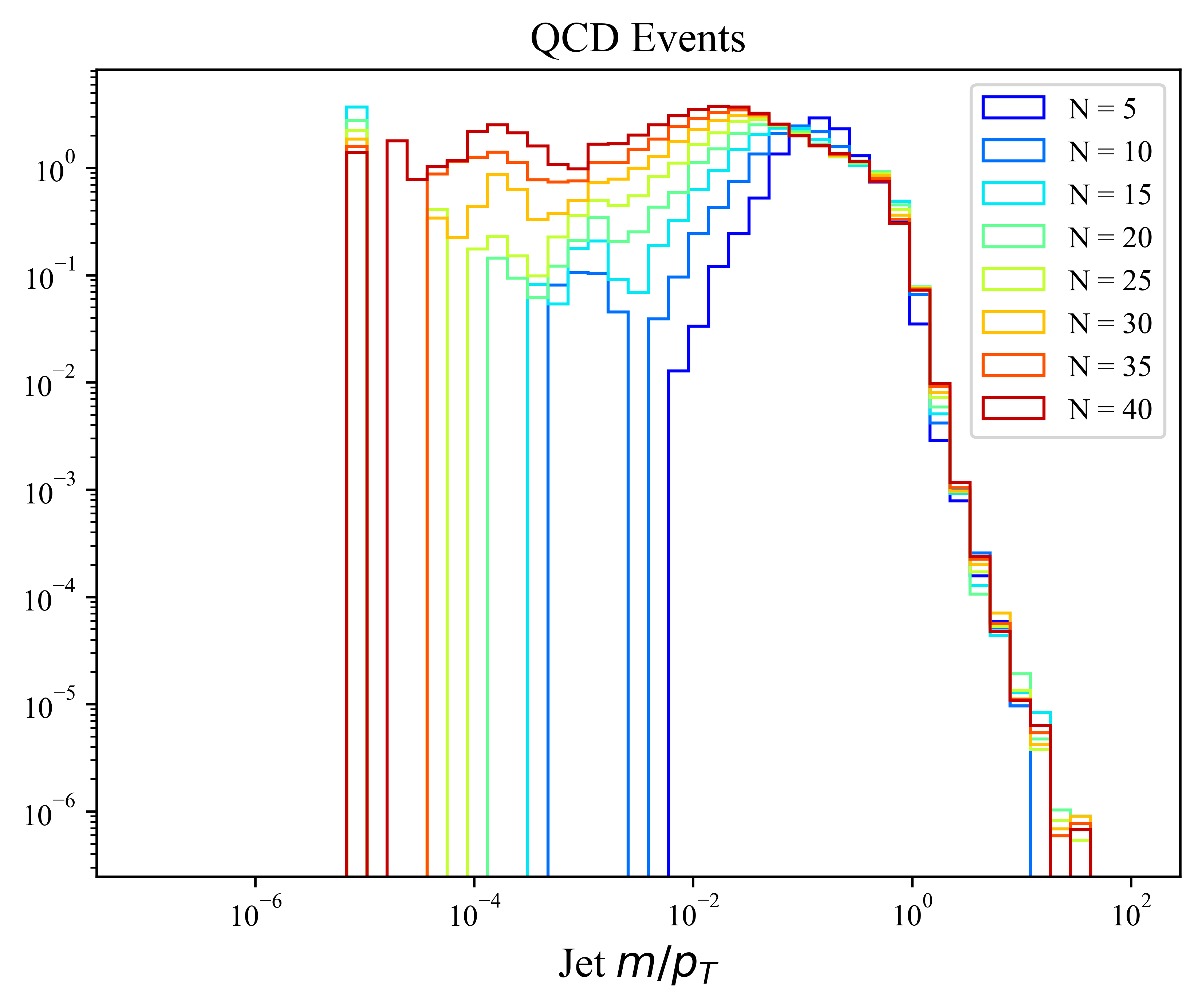}
    \includegraphics[width=0.45\textwidth]{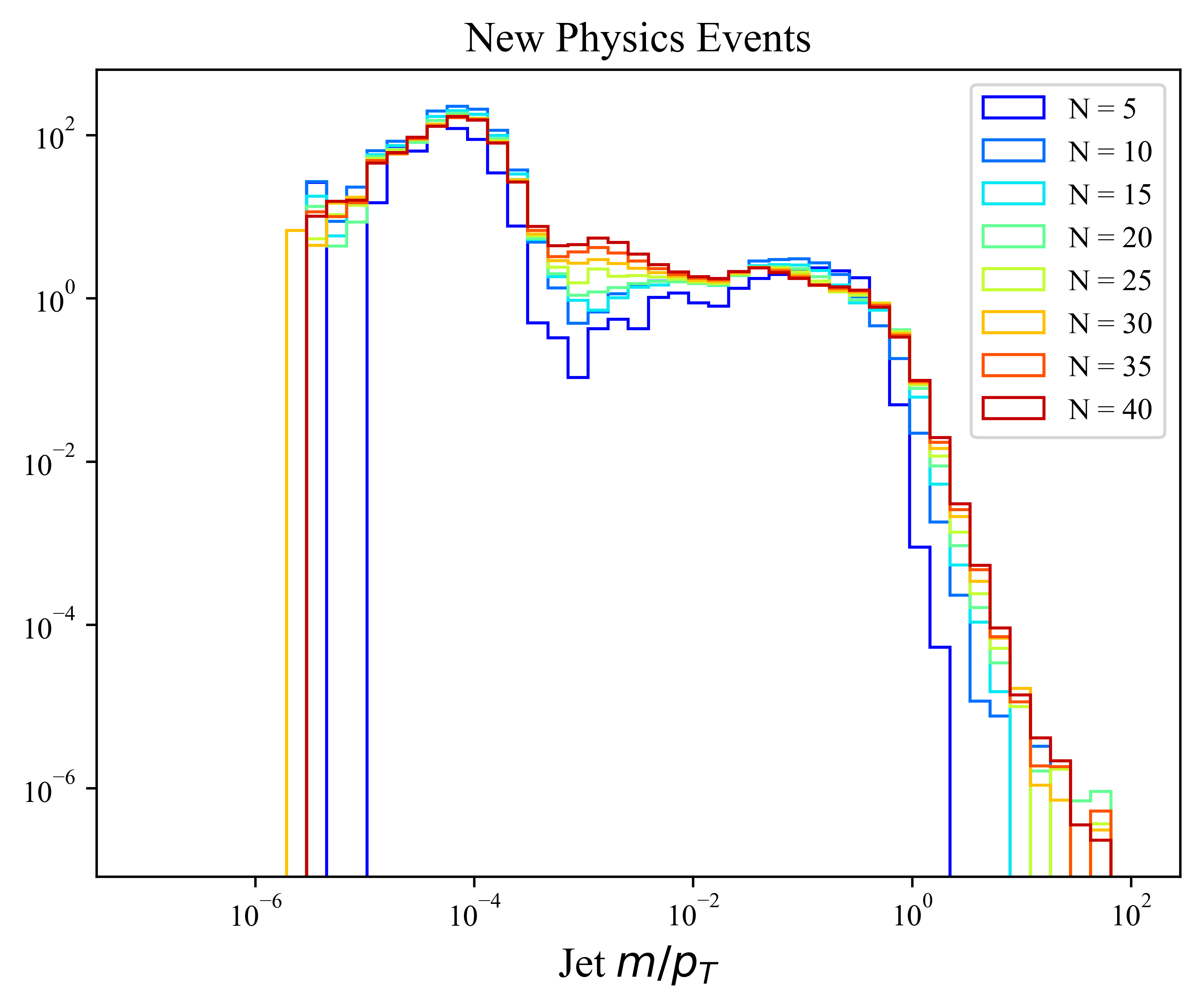}
    \caption{Histograms for each jet's $m/p_T$ (before manually setting $m$ to zero) of the QCD ({\it left}) and new physics ({\it right}) events. Note that both axes are on a log basis, and the colors denote different numbers $N$ of jets per event.}
    \label{fig:mptratio_QCDZ'}
\end{figure}

\subsubsection{Results}

Again, we first show in Fig.~\ref{fig:PStotal_QCDZv} a few sample histograms of the total phase space distances for the QCD dijet events and the new physics scenario under the set of metric definitions (\texttt{- on simplex}, \texttt{sphere:simplex = 1:1}, $(p_0-p_3)$-ordering). The difference in the histograms for the distances among various types of events is now more pronounced, hinting at potentially better classification performance as confirmed below.  

\begin{figure}
    \centering
    \includegraphics[width=\textwidth]{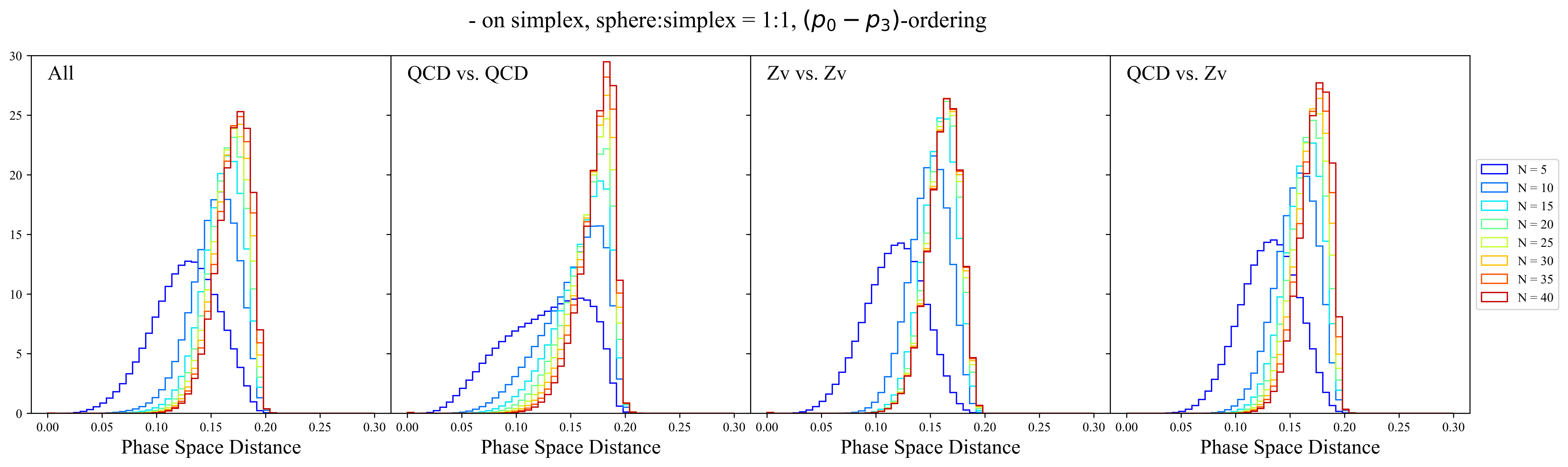}
    \caption{{\it From left to right:} Histograms of the total phase space distances between two arbitrary events in the dataset, between two QCD events, between two new physics events (indicated by ``Zv''), and between a QCD event and a new physics event. The particular set of metric definitions is chosen to be \texttt{- on simplex}, \texttt{sphere:simplex = 1:1}, and $(p_0-p_3)$-ordering. The colors denote different numbers $N$ of jets per event.}
    \label{fig:PStotal_QCDZv}
\end{figure}

Fig.~\ref{fig:QCDZvevents_l_total_SVM} shows the AUC scores of the total phase space distances. The performance is excellent across a wide range of $N$ values, with higher $N$ giving an AUC score close to 1 (i.e., 0.982). This indicates that the current signals and backgrounds are extremely well separable on the phase space manifold, and that the higher the dimension of the phase space, the better the separability. Such observations correspond well with our physics intuition, since the new physics events have a much higher multiplicity than the QCD events. Our near-perfect performance also echoes with the high AUC score of $\sim 0.930$ in Ref.~\cite{Cesarotti:2020hwb}, obtained again by the ring-like event isotropy. Both results indicate the relative ease of this particular classification problem.

For our phase space method, we again note that both definitions of the $\vec u$ coordinates give comparable performance and the sphere-to-simplex relative weighting has little impact on classification, corroborating our hypothesis that any valid phase space distance definition contains equivalent information for a reasonable classifier to utilize. Furthermore, there is very little variation in the final performance for the three particle orderings, suggesting that the local minima found by the approximation schemes are likely very close to the global minimum.     

Unlike the previous top vs.~QCD event classification task, here the performance is relatively stable with respect to the number of jets used to represent an event. Apart from the slight drop at $N \sim 10-15$ when the simplex gets a higher weight, the performance shows a modest yet steady increase when more objects are included, contrary to the trend observed in the top vs.~QCD task. This suggests that the $N$-dependence of classification is a function of the specific underlying physics, rather than the overall properties of the $N$-dependent phase space manifold.

\begin{figure}[t!]
	\centering
	\includegraphics[width=\textwidth]{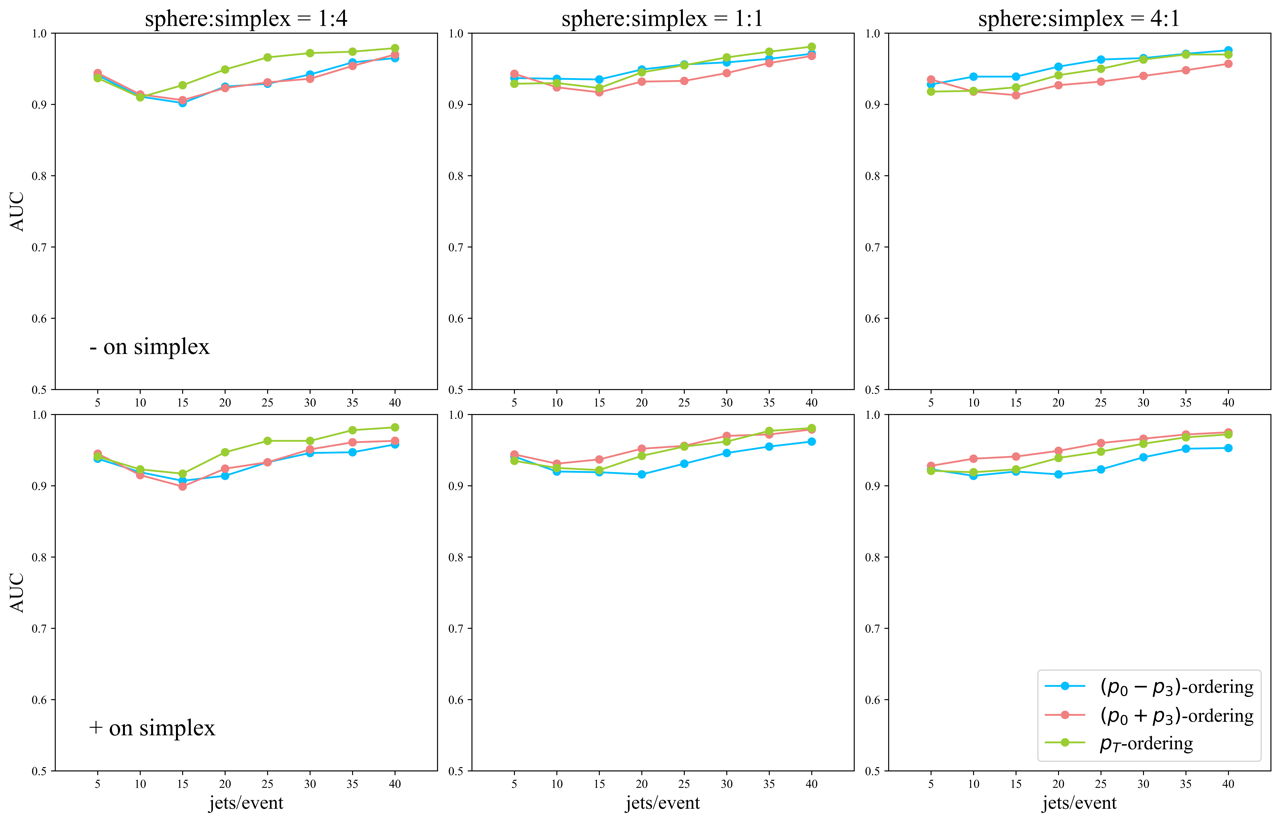}
    \caption{AUC values for QCD dijet events vs.~new physics events, using SVM as the classifier. The number of jets to represent each event (i.e., $x$-axis) is varied to be $N = 5, 10, 15, 20, 25, 30, 35, 40$. Each subplot corresponds to a distinct pair of the $\vec u$ coordinates definition coupled with a specific sphere-to-simplex weighting ratio. The blue, red, and green curves represent results with constituents ordered by decreasing $(p_0-p_3)$, $(p_0+p_3)$, and $p_T$, respectively.}
	\label{fig:QCDZvevents_l_total_SVM}
\end{figure}

Fig.~\ref{fig:QCDZvevents_Sphere11_l_simplex&sphere_SVM} studies the simplex and sphere submanifolds separately for both $\vec u$ coordinates definitions. In general, results are qualitatively the same as in the previous task. The sphere part once again proves to be more useful than the simplex, where its highest attainable AUC is less than $2 \%$ lower than the optimal AUC obtained by the total phase space distance. The three particle orderings give roughly the same performance on the sphere for all $N$, with the $p_T$-ordering being slightly better for larger $N$. 

On the simplex, we again notice the phenomenon that the particle ordering that minimizes the simplex distance performs the worst and a larger performance gap ($\sim 20\%$) now exists for different orderings. For both $(p_0+p_3)$ and $(p_0-p_3)$ orderings, the performance has a noticeable drop when $N$ approaches 15, which is around the number of hard components in the new physics events (recall $N_{Z'} = 10$). On the other hand, the $p_T$-ordering does not have this effect and its performance grows steadily with increasing $N$.

\begin{figure}[t!]
	\centering
	\includegraphics[width=0.8\textwidth]{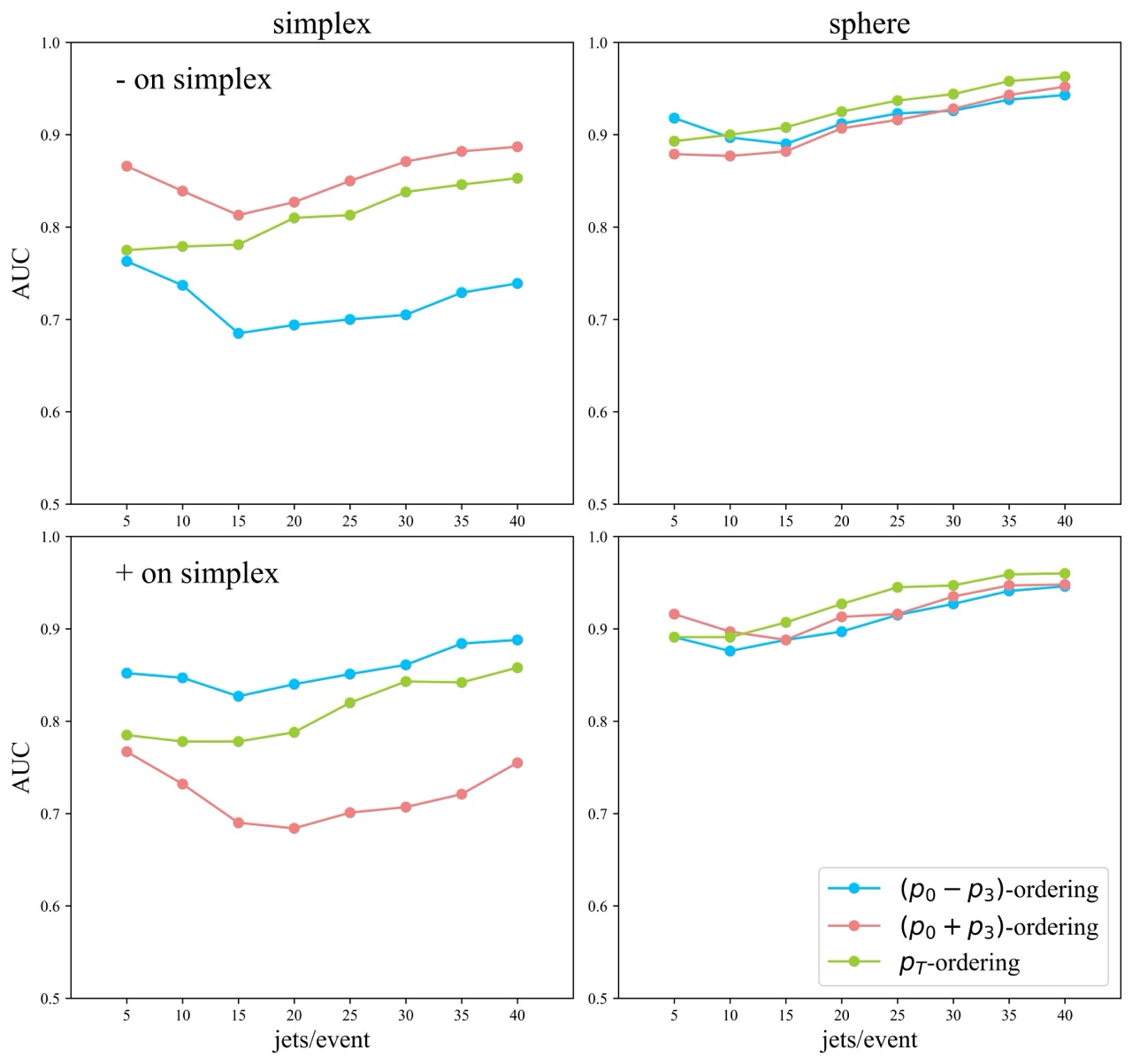}
    \caption{AUC values of QCD dijet events vs.~new physics events, where the input distance matrix is the simplex ({\it left}) or the sphere ({\it right}) part alone under both $\vec u$ coordinates definitions ({\it upper} for \texttt{- on simplex} and {\it lower} for \texttt{+ on simplex}).}
	\label{fig:QCDZvevents_Sphere11_l_simplex&sphere_SVM}
\end{figure}

\subsection{High Energy Collinear QCD vs. $W$ Boson Jets}
\label{subsec:WQCDJets}

We now focus on single jet tagging and apply the phase space formalism in the collinear limit, as described in Sec.~\ref{subsec:CollinearPS}. Proton-proton collision events are first simulated at $\sqrt{s} = 14$ TeV in \textsc{MadGraph}, with $W$ bosons being pair produced via $pp \to W^+ W^-$, gluons being generated via $q\overline{q}\to Z\to\nu\overline{\nu}g$, and quarks generated via $qg\to Z\to\nu\overline{\nu}q$. The particles are then hadronized and decayed in \textsc{Pythia} with default tuning and showering parameters. Afterwards, events are clustered into jets using \textsc{FastJet} 3.4.0 with the anti-$k_T$ inclusive algorithm and a jet radius of $R=1$. The jet transverse momentum is selected to be $p_T \in [500, 550]$ GeV and jet mass is set within the range of $m\in[60,100]$ GeV. At most two jets are kept with $|y| \leq 1.7$ and $|\phi| \leq \frac{\pi}{2}$.

To fix the number of constituents, we re-cluster each jet into $N$ subjets in \textsc{FastJet} using the exclusive $k_T$ algorithm. Only the hardest $N = 3, 4, 5, 10, 15, 20, 25, 30$ subjets are kept and are further manually set to be massless. Here the lower limit $N=3$ is chosen to reflect the two-pronged structure of a $W$ boson jet, whereas the upper limit is constrained by the inability of \textsc{FastJet} to find more than 30 subjets. We then rotate the 3-momentum of each jet so that its axis is aligned with the $+z$ direction. The jet constituents are represented in the $(E, p_x, p_y, p_z)$ coordinates, with their $E$ normalized by the total energy of the jet. 

Fig.~\ref{fig:10kWQCD_pT500to550} displays a random $W$ boson jet and a QCD jet with different $N$ subjets on the $y-\phi$ plane. Fig.~\ref{fig:mptratio_WQCD} shows the distributions of each subjet's $m/p_T$ on a log-log scale. Clearly with increasing $N$, more and more subjets have their mass becoming negligible comparing to $p_T$, again validating the masslessness assumption especially for large $N$.

\begin{figure}[t!]
    \centering
    \includegraphics[width=\textwidth]{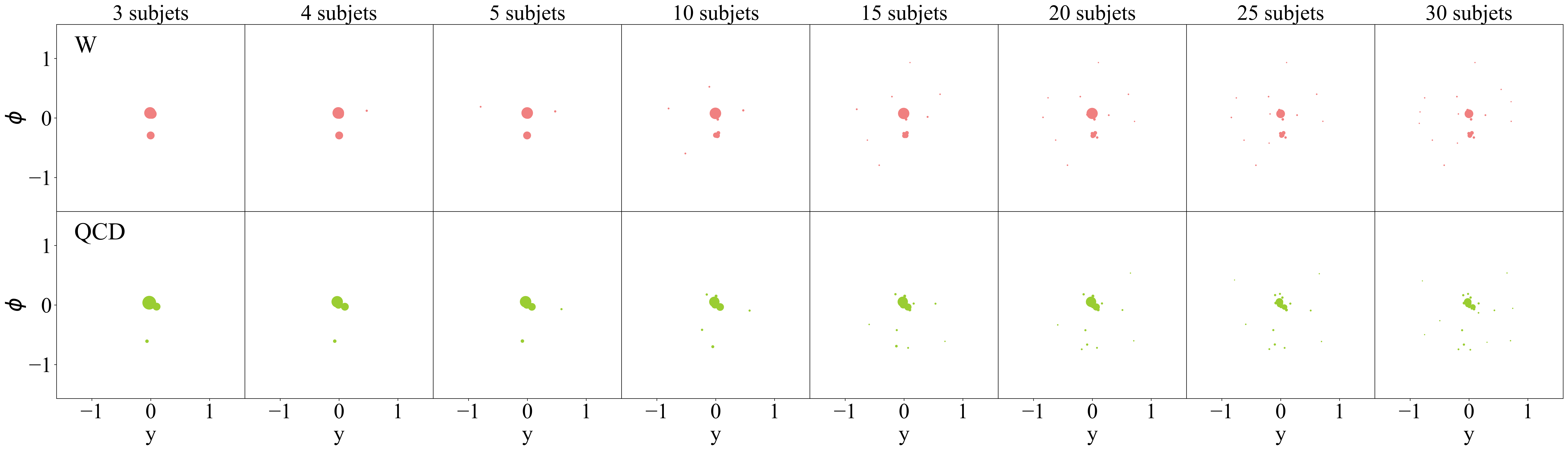}
    \caption{A $W$ (red; row 1) or QCD (green; row 2) jet as represented by $N = 3, 4, 5, 10, 15, 20, 25, 30$ exclusive subjets in the $y-\phi$ plane. The size of each dot is proportional to the $p_T$ of each subjet.}
    \label{fig:10kWQCD_pT500to550}
\end{figure}

\begin{figure}[t!]
    \centering
    \includegraphics[width=0.45\textwidth]{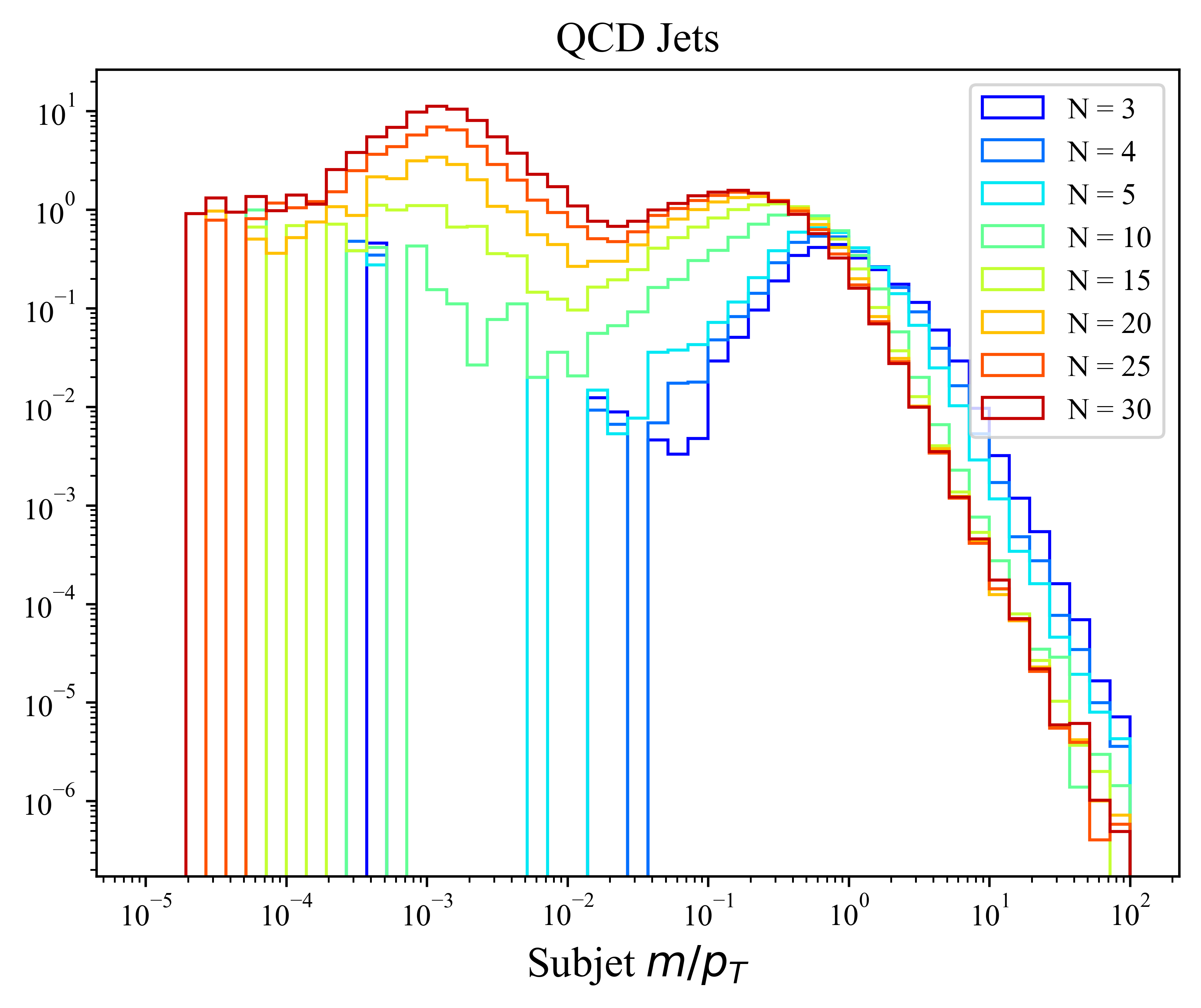}
    \includegraphics[width=0.45\textwidth]{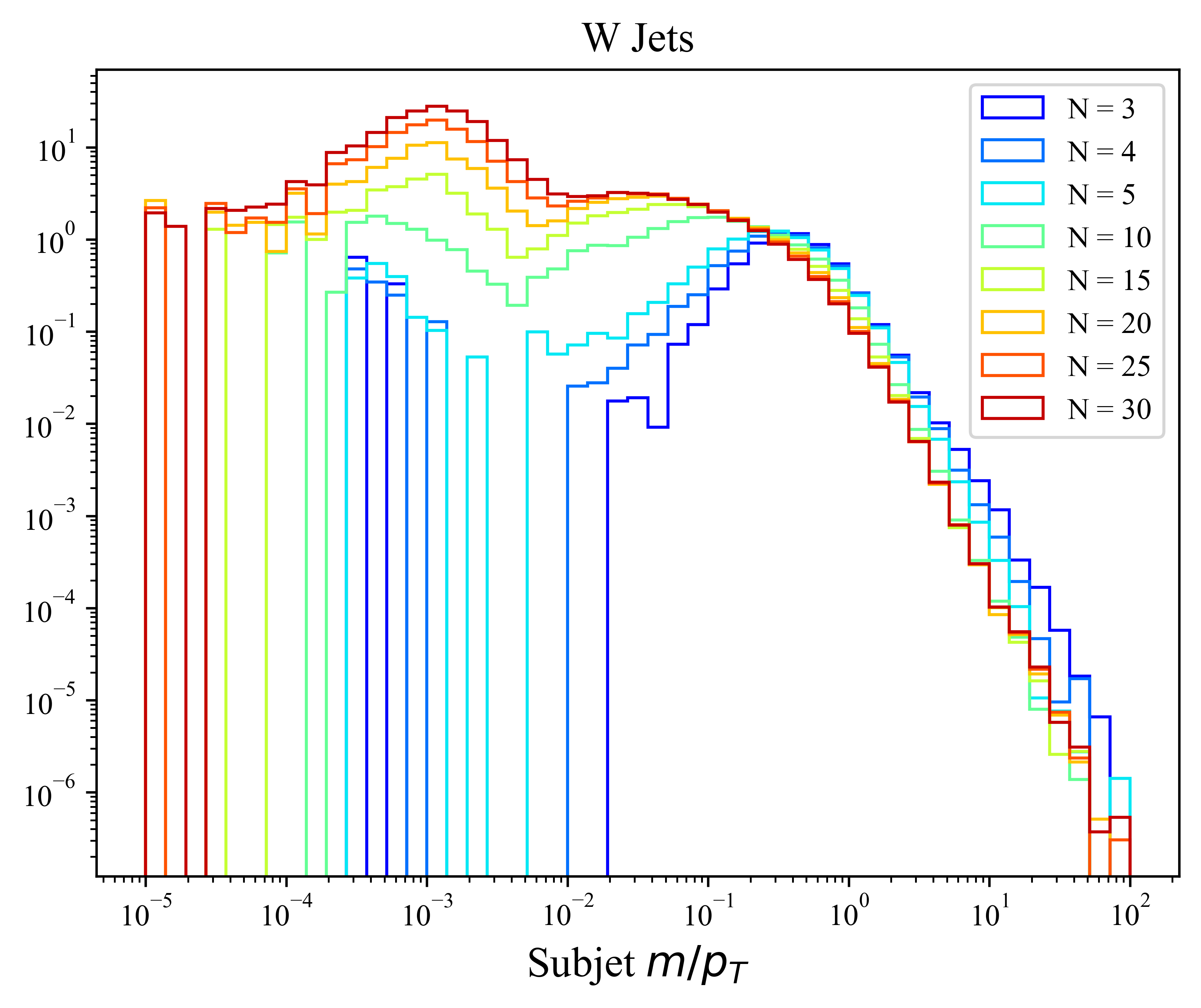}
    \caption{Histograms for each subjet's $m/p_T$ (before manually setting $m$ to zero) of the QCD ({\it left}) and $W$ ({\it right}) jets. Note that both axes are on a log basis, and the colors denote different numbers $N$ of subjets per jet.}
    \label{fig:mptratio_WQCD}
\end{figure}

\subsubsection{Results}

Fig.~\ref{fig:PStotal_WQCD} displays the histograms of the total phase space distances for the $W$ and QCD jets under the metric definitions (\texttt{- on simplex}, \texttt{sphere:simplex = 1:1}, $(p_0-p_3)$-ordering). The tagging results are displayed in Fig.~\ref{fig:QCDWjets_l_total_SVM_massless}. Overall, the performance is excellent, with the highest AUC score approaching 0.9. As before, the $\vec u$ coordinates definitions and relative simplex-to-sphere weighting have little impact on the classification, and all three particle orderings result in similar performance. The AUC score enjoys a minor increase at higher values of $N$, where the masslessness assumption holds better. In general, performance stabilizes at its peak after $N>10$.     

\begin{figure}[t!]
    \centering
    \includegraphics[width=\textwidth]{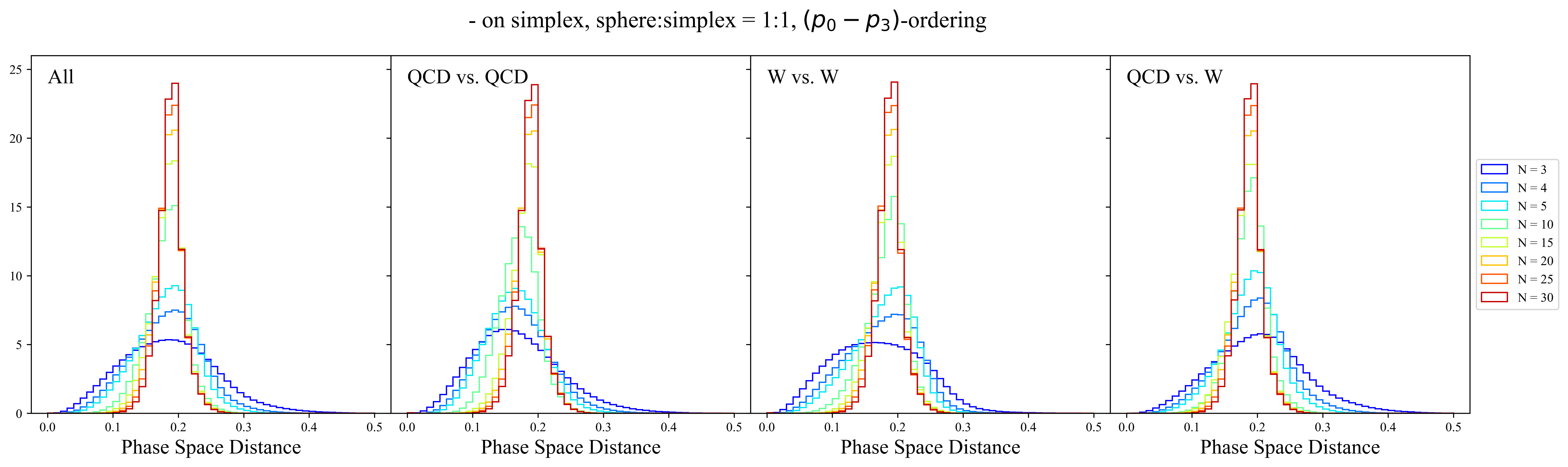}
    \caption{{\it From left to right:} Histograms of the total phase space distances between two arbitrary jets in the dataset, between two QCD jets, between two $W$ jets, and between a QCD jet and a $W$ jet. The particular set of metric definitions is chosen to be \texttt{- on simplex}, \texttt{sphere:simplex = 1:1}, and $(p_0-p_3)$-ordering. The colors denote different numbers $N$ of jets per event.}
    \label{fig:PStotal_WQCD}
\end{figure}

\begin{figure}[t!]
	\centering
	\includegraphics[width=\textwidth]{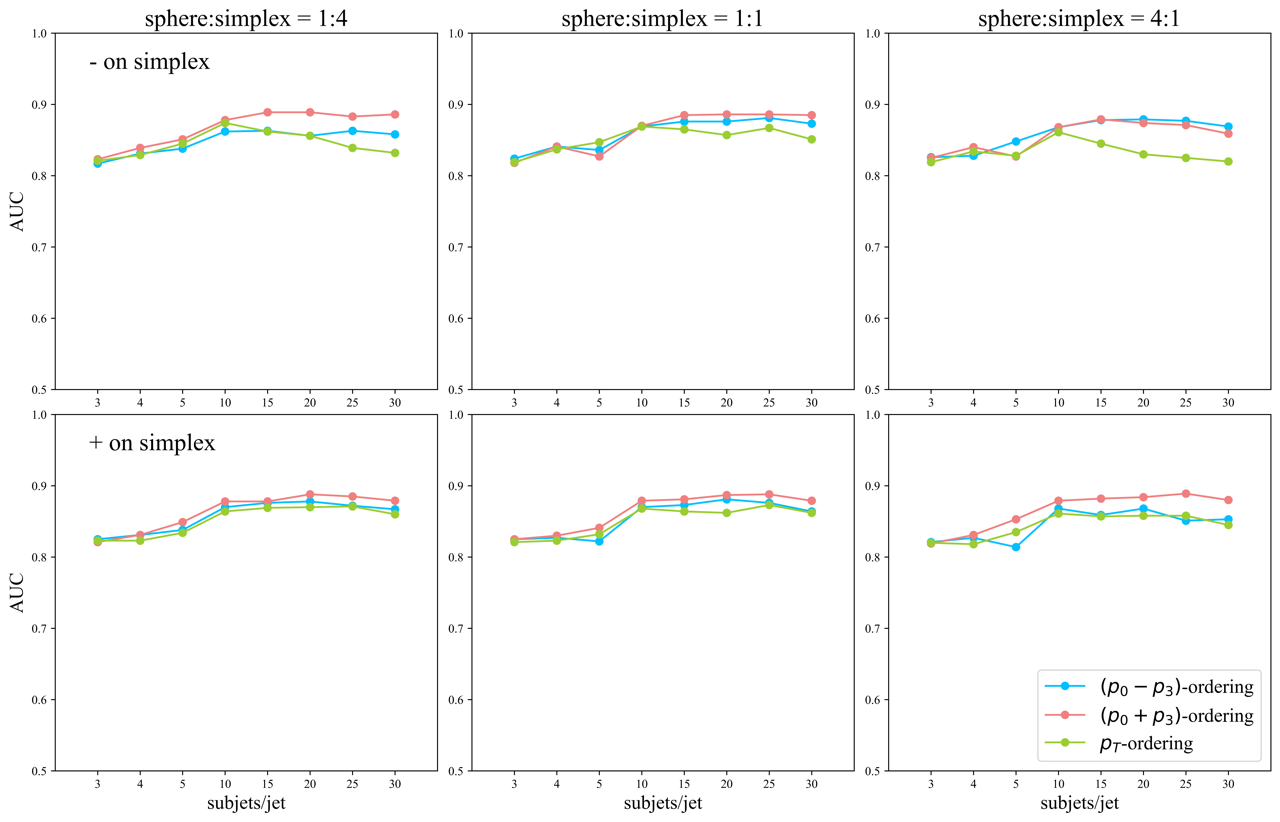}
    \caption{AUC values for $W$ vs.~QCD jets using SVM as the classifier. The number of subjets to represent each jet (i.e., the $x$-axis) is varied to be $N = 3, 4, 5, 10, 15, 20, 25, 30$. Each subplot corresponds to a distinct pair of $\vec u$ coordinates definition coupled with the sphere-to-simplex weighting ratio. The blue, red, and green curves represent results with constituents ordered by decreasing $(p_0-p_3)$, $(p_0+p_3)$, and $p_T$, respectively.}
	\label{fig:QCDWjets_l_total_SVM_massless}
\end{figure}

When zooming in to the performance of the simplex and the sphere submanifolds (see Fig.~\ref{fig:QCDWjets_Sphere11_l_simplex&sphere_SVM_massless}), we observe qualitatively similar phenomena as in the previous two event classification tasks, except that here the simplex part contains as much as, or even more, discrimination information than the sphere part, especially for $N>10$ with the best performing particle ordering. This difference highlights the importance of the simplex coordinates in the collinear limit, as now the $z$-axis is defined to be aligned with the jet axis with the value of $p_3$ carrying physical information about the jet structure. 

We also notice here that compared to the previous event classification tasks, different particle orderings now result in larger gaps in the AUC score of more than $20\%$. On the other hand, the ordering that exactly minimizes the simplex distance continues to give the worst performance for the simplex-only analysis, but is still the best performer for the sphere-only analysis. In this case, combining simplex and sphere into a total phase space distance has the distinctive advantage of stabilizing the performance for different particle orderings, as well as for a wide range of $N$ values with significant improvement at lower $N$.

\begin{figure}[t!]
	\centering
	\includegraphics[width=0.8\textwidth]{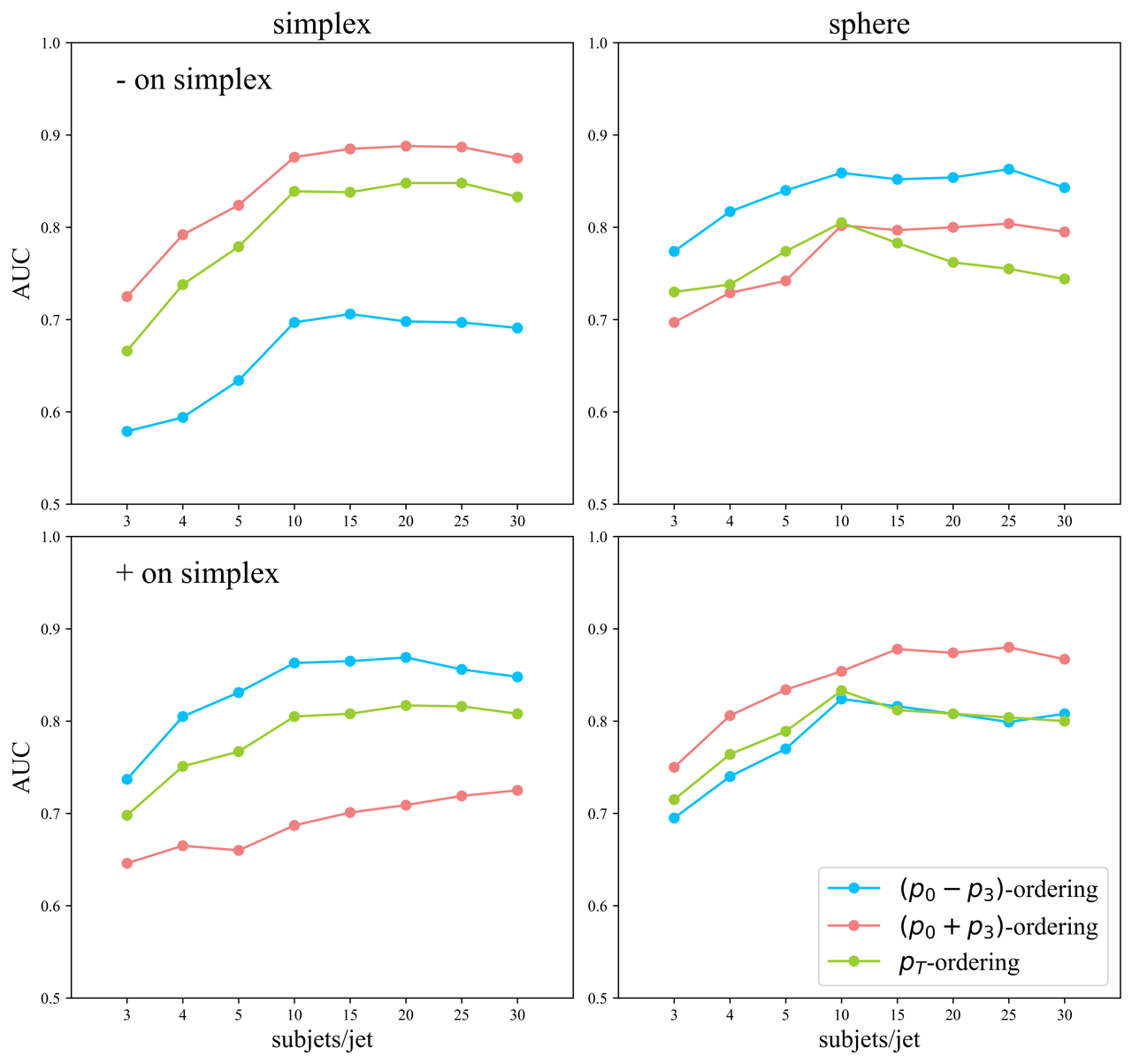}
    \caption{AUC values of $W$ vs.~QCD jets, where the input distance matrix is the simplex ({\it left}) or the sphere ({\it right}) part alone under both $\vec u$ coordinates definitions ({\it upper} for \texttt{- on simplex} and {\it lower} for \texttt{+ on simplex}).}
	\label{fig:QCDWjets_Sphere11_l_simplex&sphere_SVM_massless}
\end{figure}

\section{Conclusions and Future Directions}
\label{sec:conclusions}

The phase space manifold furnishes a natural and powerful notion of the distance between collider events. In this paper, we have developed a practical prescription for defining and computing phase space distances at hadron colliders, both for boosted objects and entire events. These distances illuminate the distinctive phase space structures of different scattering amplitudes and lead to competitive event classification when coupled to simple machine learning algorithms.

Specifically, we have found that any valid definition of the phase space distance gives rise to comparably good classification performance, reflecting the fact that the distance captures the same underlying physics. Furthermore, the particle ordering scheme to approximate the exact minimization required in the definition of the total phase space distance also bears relatively little impact on the downstream performance. This is an encouraging sign, suggesting that a future exact minimization would most likely give similar results with lower risk of qualitatively changing the conclusions drawn in the current study. 

Instead, the key factor that determines the classification performance is the number of objects used to represent an event. This is in part due to the validity of the masslessness assumption, as higher $N$ results in a lower mass of the individual clustered objects. More importantly, $N$ defines the dimensionality of the phase space manifold itself. Naturally, the performance of the distance varies with $N$ and is dependent on the specific tagging task under consideration.

There are a number of promising directions for future study. In converting explicit global coordinates on the phase space manifold to a practical metric, we have made certain choices that leave room for exploration. For example, we forego an exact (yet computationally prohibitive) minimization of the distance over particle permutations in favor of an approximate procedure that is only guaranteed to find local minima. Alternative approaches to the minimization problem (potentially using machine learning) may prove effective.

We have also fixed the dimensionality of the phase space manifold for a given ensemble of events by exclusively clustering into $N$ objects, mapping events onto an $N$-object representation with crucial influence over the effectiveness of ML-based classification. Alternative approaches to treating events of different intrinsic dimensionality may prove fruitful, and one may also want to explore the impact of different clustering algorithms (such as inclusive clustering) on the classification performance of the phase space distance.

In this work, we have only studied the simplest choices of the metrics on the simplex, the hypersphere, and their product phase space. In principle, the metrics can be generalized to be functions of the coordinates, as long as they induce the expected topology and the phase space volume element is preserved. Such generalization would grant extra freedom to adjust the distance between events according to their place in the phase space, which may lead to even better classification performance. Furthermore, the computation of the phase space distances currently makes the assumption of massless objects, which is not guaranteed to be a good approximation in all regimes. This could be circumvented by generalizing the phase space metric to massive particles, and would make the current framework more versatile.

On the practical side, while we have demonstrated the value of phase space distances in classification tasks, they are equally likely to be useful in other collider physics applications such as unsupervised anomaly detection and simulation-based inference. Further analytic and numerical study of the phase space structure of different scattering amplitudes is most likely to be productive.  

Long a backdrop to the calculation of scattering amplitudes, the geometry of phase space is beyond ready to step into the limelight.

\acknowledgments

We would like to thank Liam Brennan, Brian Henning, Joe Incandela, Tamas Vami, and Danyi Zhang for useful discussions. We extend our special thanks to the referee, who suggested several generalizations of the framework for future study. The work of TC, NC, and GK was supported in part by the U.S.\ Department of Energy under the grant DE-SC0011702 and performed in part at the Kavli Institute for Theoretical Physics, supported by the National Science Foundation under Grant No.~NSF PHY-1748958. TC is additionally supported by the U.S. Department of Energy Awards No.~DE-FOA-0002705, KA/OR55/22 (AIHEP) and No.~DE-AC02-76SF00515.

\appendix
\section{Coordinates on the Phase Space Hypersphere Submanifold}\label{app:spherematrix}

In this appendix, we derive the change of variables from $\Vec{v}$ to $\Vec{v}'$. For compactness, we will make use of index notation, where repeated indices are implicitly summed over $1,\ldots, N-1$. Consider the arguments of the delta functions in Eq.~\ref{eq:v_change_of_vars}. In order to have equality, we need
\begin{align}\label{eq:change_of_vars_condition}
    v^*_i v_i + |v_N|^2 = v'^*_k v'_k = \left( v_i^* A_{ki}^* \right) \left( A_{kj} v_j \right)
\end{align}
where the matrix $A$ represents the change of variables we wish to find. Recall that, at this point, all components of $\vec u$ are real and strictly positive, while $v_N$ is defined via
\begin{align}\label{eq:vN_implicit_def}
    v_N = -\frac{1}{u_N} \left( u_i v_i \right).
\end{align}
Plugging Eq.~\ref{eq:vN_implicit_def} into the left hand side of Eq.~\ref{eq:change_of_vars_condition}, we have
\begin{align}
    v^*_i v_i + |v_N|^2 &= v^*_i v_i + \frac{1}{u_N^2} \left( v_i^* u_i \right) \left( u_j v_j \right) \\
    &= v_i^* \delta_{ij} v_j + \frac{1}{u_N^2} \left( v_i^* u_i \right) \left( u_j v_j \right) \\
    &= v_i^* \left[ \frac{1}{u_N^2} \left(u_i u_j + u_N^2 \delta_{ij}\right) \right] v_j.
\end{align}
The quantity in brackets is essentially the square of the matrix we seek. With malice aforethought, we proceed to ``complete the square,'' making use of the $\vec u$ normalization condition, $u_i u_i = 1 - u_N^2$. We then have
\begin{align}
    &= v_i^* \left\{ \frac{1}{u_N^2} \left[u_i u_j \left( \frac{1-u_N^2}{(1+u_N)^2} + 2 \frac{u_N}{1+u_N}\right) + u_N^2 \delta_{ij}\right] \right\} v_j \\
    &= v_i^* \left[ \frac{1}{u_N^2} \left(u_i u_j \frac{u_k u_k}{(1+u_N)^2} + u_k u_j \delta_{ik} \frac{u_N}{1+u_N} + u_i u_k \delta_{kj} \frac{u_N}{1+u_N} + u_N^2 \delta_{ik}\delta_{kj}\right) \right] v_j \\
    &= v_i^* \left[ \frac{1}{u_N^2} \left( \frac{u_i u_k}{1+u_N} + u_N \delta_{ik} \right) \left( \frac{u_k u_j}{1+u_N} + u_N \delta_{kj} \right) \right] v_j,
\end{align}
from which we identify the real, symmetric transformation matrix
\begin{align}
    A_{ij} &= \frac{1}{u_N} \left( \frac{u_i u_j}{1+u_N} + u_N \delta_{ij} \right) \\
    &= \frac{1}{u_N(1-u_N^2)}\,\Hat{R}_{ij},
\end{align}
with $\Hat{R}$ given by Eq.~\ref{eq:spheremat}, as desired.

\section{Proof of Minimal Ordering for Euclidean Metric}\label{app:euclidminorder}

In this appendix, we present a simple proof regarding the permutation-invariant Euclidean metric.

\begin{lemma}
The permutation-invariant Euclidean metric distance between two sets of points $\{x_i\}_{i=1}^N,\{y_i\}_{i=1}^N$ can be expressed as
\begin{align}\label{eq:proofeucorder}
d^2\left(\{x_i\}_{i=1}^N,\{y_i\}_{i=1}^N\right) = \min_{\sigma\in S_N}\sum_{i=1}^N\left(
x_i-y_{\sigma(i)}
\right)^2 = \sum_{\substack{i=1\\x_i> x_{i+1}\\y_i> y_{i+1}}}^N(x_i-y_i)^2\,.
\end{align}
That is, one ordering of points that minimizes their distance is when they are both in descending order.
\end{lemma}

\begin{proof}
We will assume that all points are distinct, i.e., $x_i = x_j$ implies $i=j$.  With probability 1, all points are distinct in our application of this theorem to the study of phase space because the number of points is finite and drawn from a continuous probability distribution.  The case with degeneracies can be considered as well, but the main result does not change.

Without loss of generality, we can assume that the $x_i$ points are ordered, $x_1>x_2>\cdots>x_N$, while the $y_i$ points are in random order.  Expanding out the square, the metric can be expressed as
\begin{align}
 \min_{\sigma\in S_N}\sum_{i=1}^N\left(
x_i-y_{\sigma(i)}
\right)^2 = \sum_{i=1}^N \left(
x_i^2+y_i^2
\right)-2\max_{\sigma\in S_N}\sum_{i=1}^N
x_iy_{\sigma(i)}\,.
\end{align}
We only need to consider the final term on the right for the minimization of the metric. Let's consider two terms in that sum and determine if permuting them increases its value. That is, we will consider $x_i>x_j$ for $j>i$ and focus on the terms
\begin{align}\label{eq:summax}
\sum_{i=1}^N
x_iy_{\sigma(i)}\supset x_iy_{\sigma(i)}+x_jy_{\sigma(j)}\,.
\end{align}

Now, there are obviously two cases: either $y_{\sigma(i)}>y_{\sigma(j)}$ or $y_{\sigma(j)}>y_{\sigma(i)}$.  Further, we can assume, without loss of generality, that $\sigma(j) > \sigma(i)$.  First, if $y_{\sigma(i)}>y_{\sigma(j)}$, then note that
\begin{align}
(x_i-x_j)(y_{\sigma(i)}-y_{\sigma(j)})>0\,,
\end{align}
or that
\begin{align}
x_iy_{\sigma(i)} + x_jy_{\sigma(j)}>x_iy_{\sigma(j)} + x_jy_{\sigma(i)}\,.
\end{align}
That is, the points are already ordered to minimize the metric, with $i<j$ and $\sigma(i)<\sigma(j)$.  In the other case, where $y_{\sigma(j)}>y_{\sigma(i)}$, we have
\begin{align}
(x_i-x_j)(y_{\sigma(j)}-y_{\sigma(i)})>0\,,
\end{align}
or that
\begin{align}
x_iy_{\sigma(j)} + x_jy_{\sigma(i)}>x_iy_{\sigma(i)} + x_jy_{\sigma(j)}\,.
\end{align}
In this case, we can swap the indices on the $y$ points, $\sigma(i)\leftrightarrow \sigma(j)$, and after this swap, the metric is minimized with the index ordering $i< j$ and $\sigma(i)<\sigma(j)$.

We can perform this pairwise comparison of points and their contribution to the metric, swapping indices on the $y$ points as necessary, until no swap increases the value of the sum in Eq.~\ref{eq:summax}.  After this procedure, we have therefore reordered all $y$ points such that $y_i> y_j$ with $i<j$, for all $1\leq i<j\leq N$.  That is, all $y$ points have been sorted in descending order, just like the $x$ points.  Further, note that this sorting is the global minimum because any other ordering necessarily increases the metric distance.  This then proves the equality of Eq.~\ref{eq:proofeucorder}.
\end{proof}

\section{AUC Score Tables for the Classification Tasks}\label{app:tables}

In this appendix, we provide all the AUC scores for the three classification tasks discussed in the main text from Fig.~\ref{fig:QCDtevents_l_total_SVM} to Fig.~\ref{fig:QCDWjets_Sphere11_l_simplex&sphere_SVM_massless}.

\begin{landscape}
\begin{table}[ht]
    \centering
    \begin{tabular}{|c|c|c|c|c|c|c|c|c|c|c|c|}
        \hline
        Distance & $\Vec{u}$ Definition & Sphere : Simplex & Particle Ordering & $N=3$ & $N=4$ & $N=5$ & $N=10$ & $N=15$ & $N=20$ & $N=25$ & $N=30$ \\
        \hline
        \multirow{18}{*}{Total} & \multirow{9}{*}{+ on simplex} & \multirow{3}{*}{1:4} & $p_0+p_3$ & 0.828 & 0.857 & 0.858 & 0.808 & 0.785 & 0.767 & 0.697 & 0.701 \\
        \cline{4-12}
        & & & $p_T$ & 0.825 & 0.857 & 0.849 & 0.813 & 0.805 & 0.778 & 0.795 & 0.790 \\
        \cline{4-12}
        & & & $p_0-p_3$ & 0.824 & 0.852 & 0.845 & 0.808 & 0.785 & 0.779 & 0.770 & 0.755 \\
        \cline{3-12}
        & & \multirow{3}{*}{1:1} & $p_0+p_3$ & 0.809 & 0.851 & 0.855 & 0.823 & 0.814 & 0.801 & 0.785 & 0.781 \\
        \cline{4-12}
        & & & $p_T$ & 0.828 & 0.849 & 0.843 & 0.818 & 0.814 & 0.804 & 0.801 & 0.809 \\
        \cline{4-12}
        & & & $p_0-p_3$ & 0.823 & 0.848 & 0.844 & 0.809 & 0.767 & 0.761 & 0.764 & 0.759 \\
        \cline{3-12}
        & & \multirow{3}{*}{4:1} & $p_0+p_3$ & 0.813 & 0.846 & 0.845 & 0.825 & 0.824 & 0.812 & 0.800 & 0.797 \\
        \cline{4-12}
        & & & $p_T$ & 0.823 & 0.844 & 0.832 & 0.809 & 0.813 & 0.807 & 0.811 & 0.814 \\
        \cline{4-12}
        & & & $p_0-p_3$ & 0.810 & 0.838 & 0.829 & 0.786 & 0.775 & 0.756 & 0.756 & 0.744 \\
        \cline{2-12}
        & \multirow{9}{*}{- on simplex} & \multirow{3}{*}{1:4} & $p_0+p_3$ & 0.827 & 0.850 & 0.851 & 0.797 & 0.793 & 0.779 & 0.770 & 0.759 \\
        \cline{4-12}
        & & & $p_T$ & 0.821 & 0.852 & 0.853 & 0.807 & 0.798 & 0.794 & 0.782 & 0.781 \\
        \cline{4-12}
        & & & $p_0-p_3$ & 0.831 & 0.850 & 0.850 & 0.792 & 0.772 & 0.757 & 0.730 & 0.710 \\
        \cline{3-12}
        & & \multirow{3}{*}{1:1} & $p_0+p_3$ & 0.820 & 0.860 & 0.845 & 0.803 & 0.772 & 0.767 & 0.770 & 0.763 \\
        \cline{4-12}
        & & & $p_T$ & 0.821 & 0.853 & 0.848 & 0.816 & 0.809 & 0.807 & 0.805 & 0.798 \\
        \cline{4-12}
        & & & $p_0-p_3$ & 0.826 & 0.838 & 0.849 & 0.811 & 0.808 & 0.797 & 0.799 & 0.783 \\
        \cline{3-12}
        & & \multirow{3}{*}{4:1} & $p_0+p_3$ & 0.818 & 0.847 & 0.834 & 0.769 & 0.781 & 0.761 & 0.753 & 0.755 \\
        \cline{4-12}
        & & & $p_T$ & 0.832 & 0.848 & 0.834 & 0.797 & 0.805 & 0.812 & 0.808 & 0.801 \\
        \cline{4-12}
        & & & $p_0-p_3$ & 0.824 & 0.846 & 0.835 & 0.825 & 0.814 & 0.810 & 0.808 & 0.801 \\
        \hline
        \multirow{6}{*}{Simplex} & \multirow{3}{*}{+ on simplex} & \multirow{3}{*}{N/A} & $p_0+p_3$ & 0.600 & 0.652 & 0.641 & 0.620 & 0.603 & 0.606 & 0.602 & 0.595 \\
        \cline{4-12}
        & & & $p_T$ & 0.618 & 0.663 & 0.666 & 0.655 & 0.642 & 0.642 & 0.645 & 0.637 \\
        \cline{4-12}
        & & & $p_0-p_3$ & 0.691 & 0.758 & 0.763 & 0.727 & 0.726 & 0.715 & 0.687 & 0.692 \\
        \cline{2-12}
        & \multirow{3}{*}{- on simplex} & \multirow{3}{*}{N/A} & $p_0+p_3$ & 0.689 & 0.728 & 0.737 & 0.719 & 0.702 & 0.705 & 0.700 & 0.682 \\
        \cline{4-12}
        & & & $p_T$ & 0.614 & 0.649 & 0.659 & 0.647 & 0.634 & 0.624 & 0.636 & 0.639 \\
        \cline{4-12}
        & & & $p_0-p_3$ & 0.612 & 0.636 & 0.645 & 0.609 & 0.598 & 0.586 & 0.578 & 0.583 \\
        \hline
        \multirow{6}{*}{Sphere} & \multirow{3}{*}{+ on simplex} & \multirow{3}{*}{N/A} & $p_0+p_3$ & 0.774 & 0.822 & 0.829 & 0.800 & 0.757 & 0.751 & 0.726 & 0.717 \\
        \cline{4-12}
        & & & $p_T$ & 0.729 & 0.759 & 0.789 & 0.787 & 0.786 & 0.777 & 0.778 & 0.783 \\
        \cline{4-12}
        & & & $p_0-p_3$ & 0.725 & 0.758 & 0.768 & 0.756 & 0.732 & 0.726 & 0.730 & 0.715 \\
        \cline{2-12}
        & \multirow{3}{*}{- on simplex} & \multirow{3}{*}{N/A} & $p_0+p_3$ & 0.729 & 0.767 & 0.778 & 0.738 & 0.731 & 0.726 & 0.713 & 0.709 \\
        \cline{4-12}
        & & & $p_T$ & 0.731 & 0.774 & 0.805 & 0.781 & 0.772 & 0.775 & 0.782 & 0.777 \\
        \cline{4-12}
        & & & $p_0-p_3$ & 0.759 & 0.832 & 0.829 & 0.803 & 0.758 & 0.747 & 0.738 & 0.718 \\
        \hline
    \end{tabular}
    \caption{AUC scores for top vs.~QCD events}
    \label{tab:tQCD}
\end{table}
\end{landscape}

\begin{landscape}
\begin{table}[ht]
    \centering
    \begin{tabular}{|c|c|c|c|c|c|c|c|c|c|c|c|}
        \hline
        Distance & $\Vec{u}$ Definition & Sphere : Simplex & Particle Ordering & $N=5$ & $N=10$ & $N=15$ & $N=20$ & $N=25$ & $N=30$ & $N=35$ & $N=40$ \\
        \hline
        \multirow{18}{*}{Total} & \multirow{9}{*}{+ on simplex} & \multirow{3}{*}{1:4} & $p_0+p_3$ & 0.945 & 0.915 & 0.899 & 0.924 & 0.933 & 0.951 & 0.961 & 0.963 \\
        \cline{4-12}
        & & & $p_T$ & 0.941 & 0.923 & 0.917 & 0.947 & 0.963 & 0.963 & 0.978 & 0.982 \\
        \cline{4-12}
        & & & $p_0-p_3$ & 0.938 & 0.919 & 0.907 & 0.914 & 0.933 & 0.946 & 0.947 & 0.958 \\
        \cline{3-12}
        & & \multirow{3}{*}{1:1} & $p_0+p_3$ & 0.944 & 0.931 & 0.937 & 0.952 & 0.956 & 0.970 & 0.972 & 0.979 \\
        \cline{4-12}
        & & & $p_T$ & 0.935 & 0.925 & 0.922 & 0.942 & 0.955 & 0.962 & 0.977 & 0.981 \\
        \cline{4-12}
        & & & $p_0-p_3$ & 0.941 & 0.920 & 0.919 & 0.916 & 0.931 & 0.946 & 0.955 & 0.962 \\
        \cline{3-12}
        & & \multirow{3}{*}{4:1} & $p_0+p_3$ & 0.928 & 0.938 & 0.941 & 0.949 & 0.960 & 0.966 & 0.972 & 0.975 \\
        \cline{4-12}
        & & & $p_T$ & 0.921 & 0.919 & 0.923 & 0.939 & 0.948 & 0.959 & 0.968 & 0.972 \\
        \cline{4-12}
        & & & $p_0-p_3$ & 0.923 & 0.914 & 0.920 & 0.916 & 0.923 & 0.940 & 0.952 & 0.953 \\
        \cline{2-12}
        & \multirow{9}{*}{- on simplex} & \multirow{3}{*}{1:4} & $p_0+p_3$ & 0.944 & 0.914 & 0.906 & 0.923 & 0.931 & 0.936 & 0.954 & 0.970 \\
        \cline{4-12}
        & & & $p_T$ & 0.937 & 0.910 & 0.927 & 0.949 & 0.966 & 0.972 & 0.974 & 0.979 \\
        \cline{4-12}
        & & & $p_0-p_3$ & 0.941 & 0.911 & 0.902 & 0.925 & 0.929 & 0.942 & 0.959 & 0.965 \\
        \cline{3-12}
        & & \multirow{3}{*}{1:1} & $p_0+p_3$ & 0.943 & 0.924 & 0.917 & 0.932 & 0.933 & 0.944 & 0.958 & 0.968 \\
        \cline{4-12}
        & & & $p_T$ & 0.929 & 0.930 & 0.923 & 0.945 & 0.955 & 0.966 & 0.974 & 0.981 \\
        \cline{4-12}
        & & & $p_0-p_3$ & 0.937 & 0.936 & 0.935 & 0.949 & 0.956 & 0.959 & 0.964 & 0.971 \\
        \cline{3-12}
        & & \multirow{3}{*}{4:1} & $p_0+p_3$ & 0.935 & 0.918 & 0.913 & 0.927 & 0.932 & 0.940 & 0.948 & 0.957 \\
        \cline{4-12}
        & & & $p_T$ & 0.918 & 0.919 & 0.924 & 0.941 & 0.950 & 0.963 & 0.970 & 0.970 \\
        \cline{4-12}
        & & & $p_0-p_3$ & 0.928 & 0.939 & 0.939 & 0.953 & 0.963 & 0.965 & 0.971 & 0.976 \\
        \hline
        \multirow{6}{*}{Simplex} & \multirow{3}{*}{+ on simplex} & \multirow{3}{*}{N/A} & $p_0+p_3$ & 0.767 & 0.732 & 0.690 & 0.684 & 0.701 & 0.707 & 0.721 & 0.755 \\
        \cline{4-12}
        & & & $p_T$ & 0.785 & 0.778 & 0.778 & 0.788 & 0.820 & 0.843 & 0.842 & 0.858 \\
        \cline{4-12}
        & & & $p_0-p_3$ & 0.852 & 0.847 & 0.827 & 0.840 & 0.851 & 0.861 & 0.884 & 0.888 \\
        \cline{2-12}
        & \multirow{3}{*}{- on simplex} & \multirow{3}{*}{N/A} & $p_0+p_3$ & 0.866 & 0.839 & 0.813 & 0.827 & 0.850 & 0.871 & 0.882 & 0.887 \\
        \cline{4-12}
        & & & $p_T$ & 0.775 & 0.779 & 0.781 & 0.810 & 0.813 & 0.838 & 0.846 & 0.853\\
        \cline{4-12}
        & & & $p_0-p_3$ & 0.763 & 0.737 & 0.685 & 0.694 & 0.700 & 0.705 & 0.729 & 0.739 \\
        \hline
        \multirow{6}{*}{Sphere} & \multirow{3}{*}{+ on simplex} & \multirow{3}{*}{N/A} & $p_0+p_3$ & 0.916 & 0.897 & 0.888 & 0.913 & 0.916 & 0.935 & 0.947 & 0.948 \\
        \cline{4-12}
        & & & $p_T$ & 0.891 & 0.891 & 0.907 & 0.927 & 0.945 & 0.947 & 0.959 & 0.960 \\
        \cline{4-12}
        & & & $p_0-p_3$ & 0.891 & 0.876 & 0.888 & 0.897 & 0.915 & 0.927 & 0.941 & 0.946 \\
        \cline{2-12}
        & \multirow{3}{*}{- on simplex} & \multirow{3}{*}{N/A} & $p_0+p_3$ & 0.879 & 0.877 & 0.882 & 0.907 & 0.916 & 0.928 & 0.943 & 0.952 \\
        \cline{4-12}
        & & & $p_T$ & 0.893 & 0.900 & 0.908 & 0.925 & 0.937 & 0.944 & 0.958 & 0.963 \\
        \cline{4-12}
        & & & $p_0-p_3$ & 0.918 & 0.897 & 0.890 & 0.912 & 0.923 & 0.926 & 0.938 & 0.943 \\
        \hline
    \end{tabular}
    \caption{AUC scores for QCD vs.~New Physics events}
    \label{tab:tQCD}
\end{table}
\end{landscape}

\begin{landscape}
\begin{table}[ht]
    \centering
    \begin{tabular}{|c|c|c|c|c|c|c|c|c|c|c|c|}
        \hline
        Distance & $\Vec{u}$ Definition & Sphere : Simplex & Particle Ordering & $N=3$ & $N=4$ & $N=5$ & $N=10$ & $N=15$ & $N=20$ & $N=25$ & $N=30$ \\
        \hline
        \multirow{18}{*}{Total} & \multirow{9}{*}{+ on simplex} & \multirow{3}{*}{1:4} & $p_0+p_3$ & 0.821 & 0.831 & 0.849 & 0.878 & 0.878 & 0.888 & 0.885 & 0.879 \\
        \cline{4-12}
        & & & $p_T$ & 0.823 & 0.823 & 0.834 & 0.864 & 0.869 & 0.870 & 0.871 & 0.860 \\
        \cline{4-12}
        & & & $p_0-p_3$ & 0.825 & 0.831 & 0.838 & 0.870 & 0.876 & 0.878 & 0.872 & 0.867 \\
        \cline{3-12}
        & & \multirow{3}{*}{1:1} & $p_0+p_3$ & 0.825 & 0.830 & 0.841 & 0.879 & 0.881 & 0.887 & 0.888 & 0.879 \\
        \cline{4-12}
        & & & $p_T$ & 0.821 & 0.823 & 0.832 & 0.868 & 0.864 & 0.862 & 0.873 & 0.862 \\
        \cline{4-12}
        & & & $p_0-p_3$ & 0.825 & 0.827 & 0.822 & 0.870 & 0.873 & 0.881 & 0.876 & 0.864 \\
        \cline{3-12}
        & & \multirow{3}{*}{4:1} & $p_0+p_3$ & 0.819 & 0.831 & 0.853 & 0.879 & 0.882 & 0.884 & 0.889 & 0.880 \\
        \cline{4-12}
        & & & $p_T$ & 0.820 & 0.818 & 0.835 & 0.861 & 0.857 & 0.858 & 0.858 & 0.845 \\
        \cline{4-12}
        & & & $p_0-p_3$ & 0.821 & 0.827 & 0.814 & 0.868 & 0.859 & 0.868 & 0.851 & 0.853 \\
        \cline{2-12}
        & \multirow{9}{*}{- on simplex} & \multirow{3}{*}{1:4} & $p_0+p_3$ & 0.823 & 0.839 & 0.851 & 0.878 & 0.889 & 0.889 & 0.883 & 0.886 \\
        \cline{4-12}
        & & & $p_T$ & 0.821 & 0.829 & 0.845 & 0.874 & 0.862 & 0.856 & 0.839 & 0.832 \\
        \cline{4-12}
        & & & $p_0-p_3$ & 0.817 & 0.831 & 0.838 & 0.862 & 0.863 & 0.856 & 0.863 & 0.858 \\
        \cline{3-12}
        & & \multirow{3}{*}{1:1} & $p_0+p_3$ & 0.818 & 0.841 & 0.827 & 0.870 & 0.885 & 0.886 & 0.886 & 0.885 \\
        \cline{4-12}
        & & & $p_T$ & 0.819 & 0.837 & 0.847 & 0.869 & 0.865 & 0.857 & 0.867 & 0.851 \\
        \cline{4-12}
        & & & $p_0-p_3$ & 0.824 & 0.841 & 0.836 & 0.869 & 0.876 & 0.876 & 0.881 & 0.873 \\
        \cline{3-12}
        & & \multirow{3}{*}{4:1} & $p_0+p_3$ & 0.825 & 0.840 & 0.827 & 0.868 & 0.879 & 0.874 & 0.871 & 0.859 \\
        \cline{4-12}
        & & & $p_T$ & 0.819 & 0.834 & 0.828 & 0.861 & 0.845 & 0.830 & 0.825 & 0.820 \\
        \cline{4-12}
        & & & $p_0-p_3$ & 0.826 & 0.828 & 0.848 & 0.868 & 0.878 & 0.879 & 0.877 & 0.869 \\
        \hline
        \multirow{6}{*}{Simplex} & \multirow{3}{*}{+ on simplex} & \multirow{3}{*}{N/A} & $p_0+p_3$ & 0.646 & 0.665 & 0.660 & 0.687 & 0.701 & 0.709 & 0.719 & 0.725 \\
        \cline{4-12}
        & & & $p_T$ & 0.698 & 0.751 & 0.767 & 0.805 & 0.808 & 0.817 & 0.816 & 0.808 \\
        \cline{4-12}
        & & & $p_0-p_3$ & 0.737 & 0.805 & 0.831 & 0.863 & 0.865 & 0.869 & 0.856 & 0.848 \\
        \cline{2-12}
        & \multirow{3}{*}{- on simplex} & \multirow{3}{*}{N/A} & $p_0+p_3$ & 0.725 & 0.792 & 0.824 & 0.876 & 0.885 & 0.888 & 0.887 & 0.875 \\
        \cline{4-12}
        & & & $p_T$ & 0.666 & 0.738 & 0.779 & 0.839 & 0.838 & 0.848 & 0.848 & 0.833 \\
        \cline{4-12}
        & & & $p_0-p_3$ & 0.579 & 0.594 & 0.634 & 0.697 & 0.706 & 0.698 & 0.697 & 0.691 \\
        \hline
        \multirow{6}{*}{Sphere} & \multirow{3}{*}{+ on simplex} & \multirow{3}{*}{N/A} & $p_0+p_3$ & 0.750 & 0.806 & 0.834 & 0.854 & 0.878 & 0.874 & 0.880 & 0.867 \\
        \cline{4-12}
        & & & $p_T$ & 0.715 & 0.764 & 0.789 & 0.833 & 0.812 & 0.808 & 0.804 & 0.800 \\
        \cline{4-12}
        & & & $p_0-p_3$ & 0.695 & 0.740 & 0.770 & 0.824 & 0.816 & 0.808 & 0.799 & 0.808 \\
        \cline{2-12}
        & \multirow{3}{*}{- on simplex} & \multirow{3}{*}{N/A} & $p_0+p_3$ & 0.697 & 0.729 & 0.742 & 0.802 & 0.797 & 0.800 & 0.804 & 0.795 \\
        \cline{4-12}
        & & & $p_T$ & 0.730 & 0.738 & 0.774 & 0.805 & 0.783 & 0.762 & 0.755 & 0.744 \\
        \cline{4-12}
        & & & $p_0-p_3$ & 0.774 & 0.817 & 0.840 & 0.859 & 0.852 & 0.854 & 0.863 & 0.843 \\
        \hline
    \end{tabular}
    \caption{AUC scores for QCD vs.~W jets}
    \label{tab:tQCD}
\end{table}
\end{landscape}


\bibliographystyle{JHEP}
\bibliography{Draft}

\end{document}